%% file: main_manuscript.tex
\documentclass{article} % For LaTeX2e
\usepackage{iclr2027_conference,times}

\input{math_commands.tex}

\usepackage{hyperref}
\usepackage{url}
\usepackage{amsmath,amssymb,amsthm,mathtools}

\title{Mesh-Native Tokenisation: Retrofitting Foundation Models to Unstructured Data}

\author{Georg Maierhofer, Miles Cranmer \& Payel Mukhopadhyay \\
University of Cambridge \\
\texttt{\{gam37,mc2473,pm858\}@cam.ac.uk}
}

\usepackage{graphicx}
\usepackage{enumitem}
\usepackage{subcaption}
\usepackage{placeins}
\usepackage[table]{xcolor}
\usepackage{booktabs}

\newtheorem{proposition}{Proposition}

\newcommand{\bx}{\mathbf{x}}
\newcommand{\bu}{\mathbf{u}}
\newcommand{\bv}{\mathbf{v}}
\newcommand{\bp}{\mathbf{p}}
\newcommand{\bW}{\mathbf{W}}
\newcommand{\bK}{\mathbf{K}}
\newcommand{\bb}{\mathbf{b}}
\newcommand{\NN}{\mathcal{N}}
\DeclareMathOperator{\Lip}{Lip}
\newcommand{\sci}[2]{#1\times10^{#2}}
\newcommand{\pmsci}[3]{(#1\pm #2)\times10^{#3}}
\iclrfinalcopy % Uncomment for camera-ready version, but NOT for submission.
\begin{document}

\maketitle
\lhead{} % preprint: drop the "Published as a conference paper" running head

\begin{abstract}
Physics foundation models are an increasingly popular approach for building versatile emulators for spatiotemporal systems. They require tokenisation layers that lift physical field data into a latent representation suitable for dynamic processing. A typical choice for this layer is a convolutional encoder, which lifts field information in a locality-aware way and can be trained very efficiently. However, this choice cannot be applied to a large class of scientific computing problems where data are represented on unstructured meshes, for example simulations on domains with curved boundaries. Such data typically require graph neural network architectures, which are more expensive to train and, in the form currently used, encode information differently from their CNN counterparts, thus limiting the portability of physics emulators between these data regimes. In this work we introduce Mesh-Native Convolution (MNC), a locality-aware tokenisation layer that encodes unstructured mesh-based data in a way that is analogous to CNNs. This approach leads to models that can be trained once on mesh data and then deployed to arbitrary unstructured meshes, even if these are irregular, with minimal retraining. We show that this tokenisation approach can be used to retrofit grid-pretrained foundation models efficiently on unstructured meshes and provide theoretical results to confirm the consistency of this approach with classical CNN encoders. We demonstrate the favourable performance of this retrofitting approach, which in many cases reaches acceptable accuracy even in zero-shot deployment, and outperforms existing baselines after only a few epochs of fine-tuning, thus providing an efficient bridge between the two data regimes.
\end{abstract}

\section{Introduction}
\looseness=-1
Machine-learning surrogates are increasingly used to predict the evolution of complex physical systems across a wide range of scientific domains~\citep{wyder2026common,yermakov2026the,riva2026ctf4nuclear}. Many recent large pretrained models for continuum physics adopt an encoder-processor-decoder architecture, in which discretised fields are encoded into latent tokens, evolved by a large processor, and decoded back to physical space~\citep{mccabe2024multiple,herde2024poseidon,mccabe2026walrus,mukhopadhyay2026overtone}. Many of the most advanced models build on processors that originate in computer vision, thus relying on tokenisation layers, which lift physical information, often using strided convolutions~\citep{mccabe2024multiple,mukhopadhyay2026overtone}. This has proven to be a particularly efficient choice, as dense convolutions on regular grid data are highly optimised on modern accelerators~\citep{chetlur2014cudnnefficientprimitivesdeep}, but it restricts the admissible data to structured grids (cf. Fig.~\ref{fig:overview}a). However, in many scientific applications field information is not available on regular grids due to the presence of irregular or curved boundaries (engineering parts, wave modelling in coastal regions, weather forecasting) or the requirement to represent information at finer resolution in regions of the computational domain that show finer features (e.g. boundary layers in fluids, or
turbulent structures behind jet engines). In such cases data are often represented on unstructured
meshes (cf. Fig.~\ref{fig:overview}b). 

\looseness=-1
While graph neural networks (GNNs) are a natural choice for processing such data, they are more expensive to train and, in the form currently used, encode information differently from their CNN counterparts, thus limiting the portability of physics emulators between these data regimes~\citep{pfaff2021meshgraphnets}. A convolutional token has an explicit physical footprint, which mirrors the dynamics it summarises: many continuum models are formulated using local differential operators, motivating tokens whose physical footprint is controlled independently of the mesh. Standard mesh-edge message passing instead defines neighbourhoods by graph hops, whose physical extent varies with the local edge length~\citep{gilmer17a,pfaff2021meshgraphnets}, and supplying coordinates as features does not remove this dependence directly. To overcome this, we introduce Mesh-Native Convolution (MNC), a locality-aware tokenisation layer that encodes unstructured mesh-based data analogously to CNNs. MNC specifies the window geometrically and aggregates information through a continuous analogue of a convolutional kernel (Fig.~\ref{fig:MNC_overview}). This leads to a highly versatile tokenisation layer that can be used to retrofit grid-pretrained foundation models efficiently on unstructured meshes. We show that this approach is consistent with classical CNN encoders, thus permitting the direct transfer of weights from grid-pretrained models, and demonstrate the mesh-agnostic performance of this approach in zero-shot deployment. In particular, \textbf{our contributions are as follows:}

\begin{enumerate}[topsep=0pt, partopsep=0pt, itemsep=0pt, parsep=0pt, leftmargin=*]
   \item \textbf{A mesh-native convolutional tokeniser.} We introduce mesh-native convolution (MNC), a tokenisation layer for unstructured mesh-based data. This encoder is able to lift continuous physical information from unstructured data points, removing the restriction to grid-based data.
   \item \textbf{Exact consistency with CNN encoders.} We prove that MNC is consistent with classical CNN encoders, thus permitting the direct loading of weights from grid-pretrained models.
   \item \textbf{Retrofitting pretrained foundation models to meshes.} We show through extensive numerical examples that retrofitting state-of-the-art physics emulators and foundation models with MNC reaches acceptable accuracy zero-shot (cf.\ Table~\ref{tab:exp1-zeroshot}) and outperforms existing baselines after only a few epochs of fine-tuning, thus providing an efficient bridge between the two data regimes.
\end{enumerate}

\begin{figure}[h!]
\centering
\includegraphics[width=\textwidth]{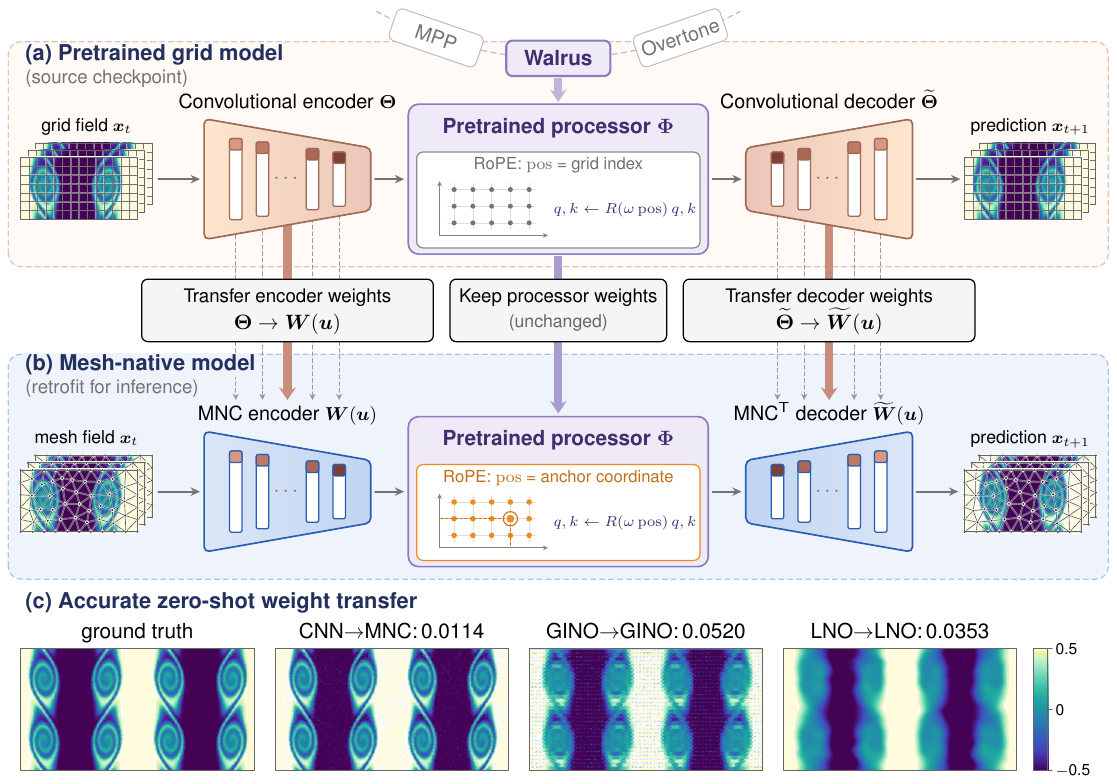}
\caption{\textbf{(a)} MNC loads grid-pretrained weights directly onto a mesh-compatible GNN architecture, for consistent retrofitting of pretrained physics emulators. \textbf{(b)} Zero-shot shear-flow single-step prediction after grid pretraining, including MNC (ours), GINO \citep{li2023gino} and LNO \citep{wang2024lno}. Values are normalised MSE for the displayed trajectories.}
\label{fig:overview}
\end{figure}

\subsection{Background and related work}

\textbf{Physics foundation models and retrofitting.}
The most powerful physics foundation models, such as MPP, Poseidon and Walrus, pretrain transformer surrogates across physical systems, all with convolutional patch embeddings on uniform grids~\citep{mccabe2024multiple,herde2024poseidon,mccabe2026walrus}. Their architectures often build on Vision Transformers (ViTs) \citep{dosovitskiy2021an}. The cost of retraining is prohibitive, and recent work by \citet{diaconu2026probabilistic} has provided initial steps towards retrofitting pretrained learned simulators post hoc with probabilistic predictions. Our approach is to retrofit the model to a new, more versatile data regime.

\looseness=-1
\textbf{Physics models on unstructured meshes.} The need to simulate physical systems on unstructured meshes is ubiquitous
in scientific computing~\citep{brenner2008mathematical,rowbottom2025gadaptivity}. MeshGraphNets~\citep{pfaff2021meshgraphnets} and related GNN simulators operate natively on meshes but encode information differently from ViT-based foundation models. Operator-learning approaches handle general geometries via coordinate transforms~\citep{li2023fourier}, attention over point sets~\citep{li2023transformer,wu2024transolver,luo2025transolver++}, neural fields~\citep{serrano2023operator}, continuous convolutions~\citep{hagnberger2025calmpde} or latent point-cloud encoders~\citep{alkin2024upt}, as do GINO~\citep{li2023gino} and LNO~\citep{wang2024lno}. Our approach is fundamentally different: we seek to encode information similarly to classical CNN architectures, exploiting physical locality as a guiding principle and thus allowing for a simple transition between grid and unstructured data regimes.

\looseness=-1
\textbf{Continuous convolutions and discretisation invariance.}
Several ingredients underlying MNC have appeared separately in prior work. SplineCNN parametrises geometric convolutions using continuous B-spline kernels~\citep{fey2018splinecnn}, while discretisation-invariant networks show how pretrained discrete convolutional filters can be interpolated to continuous kernels and recovered exactly on their native grid~\citep{wang2023discretization}. MNC combines these ideas into a mesh-native, multilevel encoder-decoder architecture designed specifically to retrofit pretrained grid-based physics surrogates, including their positional interface. Its operator interpretation follows the numerical-discretisation viewpoint of~\citet{boulle2024guide,kovachki2023neural}: convolution is treated as the discretisation of a local integral operator, with MNC's occupancy normalisation correcting for the non-uniform sampling induced by irregular meshes. Through this, our work provides the central insight that connects this encoder layer with physical consistency.

\section{Mesh-native tokenisation and weight transfer}
\label{sec:method}
We begin by introducing mesh-native convolution, the central ingredient of our retrofitting approach.

\textbf{Setup.} We assume the physical field is sampled at positions $\bp_j\in\Omega\subset\mathbb R^2$, with features $\bx_j\in\mathbb R^{F_{\rm in}}$. On a grid of spacing $(h_x,h_y)$, a standard convolutional layer with kernel size $k_h\times k_w$ has offsets
\[
\mathcal O=\{-L_x,-L_x+h_x,\ldots,L_x\}\times
\{-L_y,-L_y+h_y,\ldots,L_y\},\qquad
L_x=\tfrac{k_w-1}{2}h_x,\; L_y=\tfrac{k_h-1}{2}h_y.
\]
Writing $\bp_{\hat v}$ for the window centre, a strided convolution can be written as
\begin{equation}
\bx'_{\hat v}=\sum_{\bu\in\mathcal O}\Theta[\bu]\,
\bx_{j(\bp_{\hat v}-\bu)}+\bb,
\label{eq:cnn}
\end{equation}
where $\Theta[\bu]\in\mathbb R^{F_{\rm out}\times F_{\rm in}}$, $j(\bp)$ is the node at $\bp$, and $P=k_hk_w$. MNC follows this form closely.

\begin{figure}[h!]
\centering
\includegraphics[width=\textwidth]{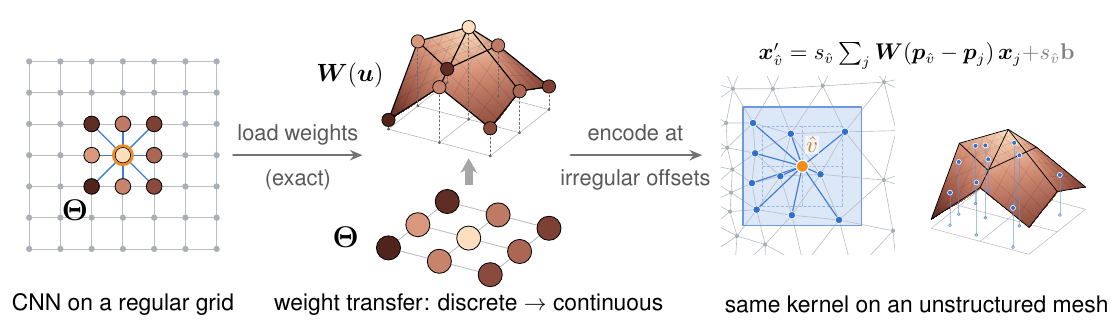}
\caption{Schematic overview of the MNC layer. MNC loads a pretrained CNN's weights into a continuous spline and evaluates the same kernel at physical offsets on an unstructured mesh.}
\label{fig:MNC_overview}
\end{figure}
\textbf{Continuous spline kernel.} Let $\phi(t)=(1-|t|)_+$. Then we can define a continuous map
\begin{equation}
\bW(\bu)=\sum_{\bu'\in\mathcal O^+}\bK[\bu']\,
\phi\!\left(\frac{u_x-u'_x}{h_x}\right)
\phi\!\left(\frac{u_y-u'_y}{h_y}\right),
\label{eq:kernel}
\end{equation}
which provides a continuous extension of the discrete kernel $\bK$ (cf. Fig.~\ref{fig:MNC_overview}). This continuous kernel forms the basis of MNC's gather and scatter operations, which are used to encode and decode information on unstructured meshes. In particular, the kernel weights $\bK$ will become learnable parameters of the MNC layer, which can also be initialised from pretrained CNN weights. To preserve physical weighting we normalise the MNC accumulation using the occupancy profile
\begin{equation}
\omega(\bu)=\sum_{\bu'\in\mathcal O}
\phi\!\left(\frac{u_x-u'_x}{h_x}\right)
\phi\!\left(\frac{u_y-u'_y}{h_y}\right),
\label{eq:occupancy}
\end{equation}\looseness=-1
which takes values in $[0,1]$ and measures whether a node falls within the (graded) support of $\bW$.

\begin{minipage}[c]{0.6\textwidth}
\textbf{Mesh hierarchy.} Much like traditional convolutional layers, we introduce the action of a single MNC layer as a ``gathering'' of information in the sense that it aggregates information from neighbouring physical locations weighted by distance, enforcing locality by virtue of the finite support of the kernel $\bK$. For this we need to assign a physical footprint to each gathering operation. This can be seen as the analogue of the position of a single convolutional window on a grid, and results in a natural definition of a mesh hierarchy of virtual anchors (cf. Fig.~\ref{fig:hierarchy}). Stacking several such interconnected meshes forms the encoder block. Analogously to transposed convolution, the action of the decoder is to distribute information across nodes and, as such, the decoder traverses the levels in reverse.
\end{minipage}\hfill
\begin{minipage}[c]{0.38\textwidth}
\centering\includegraphics[width=0.9\linewidth]{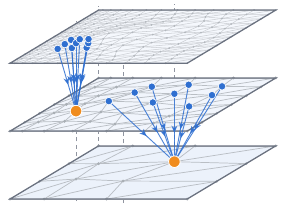}
\captionof{figure}{The mesh hierarchy, finest level at the top. Encoder steps gather onto virtual anchors; decoder steps reuse the same edges in reverse.}
\label{fig:hierarchy}
\end{minipage}

\textbf{The MNC layer and its adjoint.} We are now in a position to define the functional form of a single MNC layer, which joins each anchor $\hat v$ to all nodes in its geometric support (cf. Fig.~\ref{fig:hierarchy}) and computes
\begin{equation}
\bx'_{\hat v}=s_{\hat v}\left[\sum_{j\in\NN(\hat v)}
\bW(\bp_{\hat v}-\bp_j)\bx_j+\bb\right],\qquad
s_{\hat v}=\frac{P}{\max\{\sum_{j\in\NN(\hat v)}\omega(\bp_{\hat v}-\bp_j),1\}},
\label{eq:enc}
\end{equation}
where $P=k_hk_w$. The scaling factor $s_{\hat v}$ ensures that the total occupancy of the window is correctly normalised. The decoder reverses the same edges and uses its own kernel $\widetilde{\bW}$:
\begin{equation}
\bx'_j=t_j\left[\sum_{\hat s:j\in\NN(\hat s)}
\widetilde{\bW}(\bp_j-\bp_{\hat s})\bx_{\hat s}+\widetilde\bb\right],
\quad
t_j=\frac{1}{\max\{\sum_{\hat s:j\in\NN(\hat s)}\omega(\bp_j-\bp_{\hat s}),1\}}.
\label{eq:dec}
\end{equation}

This construction is a continuous analogue of a convolutional layer, and we now show that it is consistent with the discrete convolutional layer \eqref{eq:cnn}. The proof is provided in Appendix~\ref{app:exact}.
\begin{proposition}[Exact grid reduction]
\label{prop:exact}\looseness=-1
Suppose the input nodes $\{\bp_j\}$ are exactly the points of the source grid of spacing $(h_x,h_y)$, the physical kernel spacings equal the grid spacings, and the anchors and padding reproduce the output windows of a unit-dilation \texttt{Conv2d}. Load $\bK[\bu']=\Theta[\bu']$ for $\bu'\in\mathcal O$ and realise padding as zero-valued nodes with their occupancy included. Then \eqref{eq:enc} equals \eqref{eq:cnn} at every anchor node. Similarly, if $\widetilde\bK$ is loaded from the kernel weights of a \texttt{ConvTranspose2d} with unit dilation, stride equal to its kernel size and zero output padding, then \eqref{eq:dec} equals that \texttt{ConvTranspose2d} at every grid node.
\end{proposition}

\textbf{Preserving the processor interface.} For transformer-based processors, the tokenisation layer must also provide a positional embedding. In MNC, the source processor encodes position with axial rotary position embedding (RoPE)~\citep{su2021roformer}: on an $n_x\times n_y$ token grid, every query and key channel pair $r$ is rotated by the angle $\omega_{a,r}\,\mathrm{pos}_a$, where $\omega_{a,r}$ are fixed per-axis frequencies and $\mathrm{pos}_a=\mathrm{linspace}(-1,1,n_a)$ is the token's normalised index along axis $a$. The MNC encoder returns a flat list of anchors rather than a grid, so we evaluate the same embedding at the anchor coordinates,
\begin{equation}
\mathrm{pos}_a(\bp)=\frac{2\,(p_a-p_a^{\min})}{p_a^{\max}-p_a^{\min}}-1,\qquad
R_{a,r}(\bp)=
\begin{pmatrix}\cos(\omega_{a,r}\mathrm{pos}_a)&-\sin(\omega_{a,r}\mathrm{pos}_a)\\
\sin(\omega_{a,r}\mathrm{pos}_a)&\cos(\omega_{a,r}\mathrm{pos}_a)\end{pmatrix},
\label{eq:rope}
\end{equation}
with $p_a^{\min},p_a^{\max}$ the outermost anchor coordinates and the source frequencies, channel pairing and temporal embedding unchanged. On a regular grid this construction reduces to standard RoPE, and the processor sees the same angles for the same physical locations. As a result, this spatial encoding lends itself naturally to the retrofitting of pretrained models to unstructured meshes, and we will demonstrate this favourable property in \S\ref{sec:forecasting}. Figure~\ref{fig:rope} in Appendix~\ref{app:exact} illustrates the coordinate form and why the stock axial embedding cannot be applied to a flat list of anchors.
\section{A finite-window operator interpretation for MNC}
\label{sec:theory}
A key feature of MNC is its ability to handle unstructured data. It turns out that, beyond mere compatibility with pretrained CNN kernels, MNC is also invariant under the choice of discretisation of a given field, at least in the limit of fine sampling. This provides additional justification for the use of MNC as a physics encoder architecture, as it is consistent with the underlying continuous operator that the CNN kernel approximates. In this section we formalise this property.

\textbf{Setup.} Let the domain, $\Omega$, be bounded. We fix the spline support $S=[-R_x,R_x]\times[-R_y,R_y]$ and refine the kernel spacing so that $h_x,h_y\to0$. Let $a_h=h_xh_y$, $P_h=k_hk_w$, and choose
\begin{equation}
\bK_h[\bu']=a_h\kappa(\bu'),\qquad
(\mathcal Ku)(\bp)=\int_\Omega\kappa(\bp-y)u(y)\,dy,
\label{eq:continuum-kernel}
\end{equation}
where $\kappa\in W^{1,\infty}(\mathbb R^2;\mathbb R^{F_{\rm out}\times F_{\rm in}})$ is a Lipschitz kernel with $\operatorname{supp}\kappa\subseteq S$  and $u\in C(\overline\Omega;\mathbb R^{F_{\rm in}})$. %Then $\bW_h/a_h$ is its bilinear interpolant, and $P_ha_h\to|S|$.

\begin{proposition}[Continuum limit and discretisation invariance]
\label{prop:continuum}
Suppose the mesh nodes satisfy $c_n\sum_j\delta_{\bp_j^{(n)}}\rightharpoonup\rho(y)\,dy$, with $c_n\downarrow0$ and $\rho\in L^1(\Omega)$ positive almost everywhere. For each fixed anchor whose full window lies inside $\Omega$, MNC with zero bias, $\bb=\mathbf0$, applied to $\bx_j=u(\bp_j^{(n)})$ has the joint limit
\begin{equation}
\operatorname{MNC}_{h,n}u(\bp)\longrightarrow
|S|\,
\frac{\int_\Omega\kappa(\bp-y)u(y)\rho(y)\,dy}
{\int_{\bp-S}\rho(y)\,dy}.
\label{eq:density-limit}
\end{equation}
Thus two equidistributed refinements of the same field, with the same limiting kernel, have asymptotically identical outputs at common anchors (Fig.~\ref{fig:limit}a).
\end{proposition}

The proof of this result is included in Appendix~\ref{app:continuum}. Essentially, this means that, for sufficiently fine meshes, the output of MNC is independent of the specific mesh used, as long as the mesh is equidistributed and the kernel is appropriately defined. This means that MNC truly provides a physically consistent representation of the underlying continuous process. This is a specific property of our MNC construction and does not hold for arbitrary GNN architectures: in particular, for Graph Convolutional Networks (GCNs)~\citep{kipf2017semisupervised} and their variants, one can easily show that the output of the GCN depends on the local connectivity of the mesh, and thus on the specific triangulation used. This is a key difference between MNC and standard GNNs, and highlights the importance of our construction for retrofitting pretrained models to unstructured meshes.

\begin{figure}[h!]
\centering
\includegraphics[width=\textwidth]{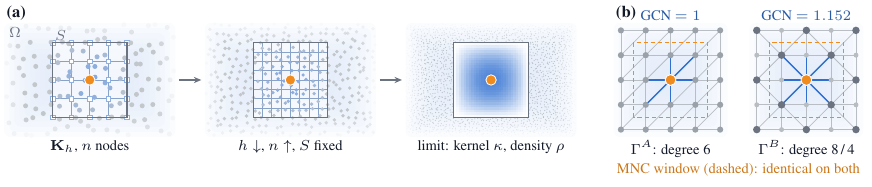}
\caption{\looseness=-1  Discretisation invariance of MNC. \textbf{(a)} Refining kernel spacing and mesh together (Proposition~\ref{prop:continuum}). In the limit MNC corresponds to the continuous kernel $\kappa$ and the mesh enters only through its sampling density $\rho$. \textbf{(b)} Two triangulations of the same vertices. The GCN neighbourhood of the origin is shown in blue. On the constant field the degree-normalised GCN returns $1$ and $1.152$ at the origin~\eqref{eq:gcn-main}, whereas the MNC window (dashed) yields identical outputs.}
\label{fig:limit}
\end{figure}

\textbf{A mesh-edge GNN need not share this property.} Consider the GCN~\citep{kipf2017semisupervised} with symmetric degree normalisation,
$G_\Gamma(\bx)=\operatorname{ReLU}(\widetilde D^{-1/2}(A+I)\widetilde D^{-1/2}\bx)$,
with scalar weight one. There are two shape-regular triangulations $\Gamma^A_\eta,\Gamma^B_\eta$ of the \emph{same} square-grid vertices of mesh width $\eta$ for which the constant field $u\equiv1$ produces, at their common origin,
\begin{equation}
G_{\Gamma^A_\eta}(\mathbf1)_0=1,
\qquad G_{\Gamma^B_\eta}(\mathbf1)_0=\frac59+\frac{4}{3\sqrt5}=1.15184\ldots
\label{eq:gcn-main}
\end{equation}
for every $\eta=1/n$ with $n\geq3$: local degrees are six in one triangulation and eight/four in the other (see Appendix~\ref{app:gcn}). MNC agrees exactly on this pair (Fig.~\ref{fig:limit}b). Symmetric normalisation does not preserve constants on irregular graphs. Mean aggregation $\widetilde D^{-1}(A+I)$ would return $1$ on both meshes, but any fixed number of hops has a footprint of size $O(\eta)$, so such a layer tends to the pointwise value $u(\bp)$ rather than to a finite-window operator.
\section{Experiments}
\label{sec:experiments}
\label{sec:forecasting}
\looseness=-1
We now test MNC in retrofitting experiments of physics emulators for dynamical problems against four mesh-native architectures that span the state of the art in encoder families for unstructured data.

\textbf{Setup.} For this evaluation we use three datasets from The Well~\citep{ohana2024well}, the two-dimensional turbulent radiative layer (TRL2D, $128\times384$), shear flow ($256\times512$) and active matter, and the MeshGraphNets cylinder-flow benchmark~\citep{pfaff2021meshgraphnets}, whose trajectories are simulated on per-trajectory unstructured meshes of about 1.9k nodes around an obstacle and have no grid representation. For the Well datasets the grid is the pretraining domain and the target is an irregular Delaunay triangulation of the same domain at matched resolution. Each dataset comes as a sequence of time-trajectories and trained models are evaluated on held-out test trajectories. Full details of the datasets and meshes used are provided in Appendix~\ref{app:datasets}.

The models are asked to predict the next step of the physical field given six input frames, and are evaluated on their ability to roll out predictions over ten steps. We report variance-normalised RMSE (VRMSE) on the physical fields, node-area weighted and averaged over non-constant channels, for the next step and averaged over a ten-step autoregressive rollout; mean~$\pm$ standard deviation is over three seeds. Full details of the training and evaluation protocol are provided in Appendix~\ref{app:protocol}.

\textbf{MNC retrofit.} We consider the MNC retrofit of two types of physics emulators: (i) in Experiments 1 and 2 (\S\ref{sec:zeroshot}-\S\ref{sec:finetune}) we retrofit a grid-pretrained Overtone-family model~\citep{mukhopadhyay2026overtone} and test on an irregular triangulation of its own dataset. (ii) in Experiment 3 (\S\ref{sec:multi-physics}) we retrofit a released multi-physics foundation model~\citep{mccabe2026walrus} and apply it to the MeshGraphNets cylinder-flow benchmark. Full model specifications can be found in Appendix~\ref{app:models}. In the following we note that for the MNC retrofitted models \textit{any pretraining is always done on grid-based data with a standard CNN encoder/decoder pair} whose weights are then directly loaded onto the MNC encoder/decoder pair, and the processor is left unchanged. \textit{Only fine-tuning is performed using the MNC encoder/decoder pair} on the target mesh data.

\textbf{Baselines.} We compare against four state-of-the-art baseline physics emulators that are able to treat unstructured data. The important difference is that our MNC layer allows for retrofitting and grid pretraining while the baselines inherently only treat mesh data, gaining no benefit from a grid-based data format. Our baselines are the neural operators GINO~\citep{li2023gino} and LNO~\citep{wang2024lno}, whose encoders act on point sets, the point-cloud transformer Transolver++~\citep{luo2025transolver++}, and the mesh-edge graph network MeshGraphNets~\citep{pfaff2021meshgraphnets}. Where available, we also performed a processor-free comparison of the tokenisers alone, which is included in Appendix~\ref{sec:tokeniser-experiments} and shows the favourable transfer properties of the MNC layer even in the absence of a processor.

\subsection{Experiment 1: zero-shot transfer onto an irregular mesh}
\label{sec:zeroshot}
\looseness=-1
Our first experiment evaluates the zero-shot transfer of grid-pretrained models to irregular mesh representations. The results are shown in Table~\ref{tab:exp1-zeroshot}, which evaluates each grid-pretrained model on an irregular mesh representation of the unseen test dataset with similar physical properties (cf.\ Appendix~\ref{app:datasets}). The Overtone+MNC model was pretrained using a CNN encoder/decoder pair on the grid representation of the dataset (using the model \citep{mukhopadhyay2026overtone}). These pretrained weights were loaded onto the MNC encoder/decoder pair for evaluation on the irregular mesh. The baselines were trained from scratch on the grid mesh representation of the dataset. The results show that the MNC retrofit retains the most pretrained skill on all three datasets and under both metrics, outperforming all baselines in zero-shot transfer. An example of the zero-shot transfer for shear flow is shown in Fig.~\ref{fig:overview}b, and the error fields in next-step prediction for TRL2D are shown in Fig.~\ref{fig:trl2d-zeroshot-onestep}. Here the fainter error field of the Overtone+MNC model means that it retains more of the pretrained skill.

\begin{table}[h!]
  \centering
  \caption{Zero-shot transfer onto the irregular mesh (Experiment 1): next-step and ten-step rollout
  VRMSE, in units of $10^{-2}$. Mean $\pm$ standard deviation over three
  pretraining seeds.}
  \label{tab:exp1-zeroshot}
  \small
  \setlength{\tabcolsep}{4pt}
  \begin{tabular}{l>{\columncolor{gray!12}}ccccc}
    \toprule
    & Overtone+MNC & GINO & LNO & Transolver++ & MeshGraphNets \\
           \multicolumn{1}{l}{\emph{Next-step VRMSE}} &&\\
    TRL2D         & {\boldmath$26.25 \pm 0.45$} & $29.55 \pm 0.10$ & $48.78 \pm 3.67$ & $40.46 \pm 0.28$ & $45.07 \pm 0.99$ \\
    Active matter & {\boldmath$4.88 \pm 0.12$} & $11.54 \pm 0.20$ & $8.00 \pm 0.45$ & $12.44 \pm 0.07$ & $19.58 \pm 0.18$ \\
    Shear flow    & $1.42 \pm 0.03$ & $1.82 \pm 0.00$ & $3.85 \pm 0.04$ & {\boldmath$1.04 \pm 0.22$} & $2.55 \pm 0.02$ \\
    \midrule
    \multicolumn{1}{l}{\emph{10-step rollout VRMSE}}\!\!\!\! &&\\
    TRL2D         & {\boldmath$52.27 \pm 0.98$} & $63.64 \pm 0.85$ & $85.03 \pm 7.48$ & $79.93 \pm 1.90$ & $85.04 \pm 2.62$ \\
    Active matter & {\boldmath$26.05 \pm 0.61$} & $73.59 \pm 4.04$ & $54.04 \pm 1.73$ & $84.74 \pm 0.94$ & $113.84 \pm 10.00$ \\
    Shear flow    & {\boldmath$8.19 \pm 0.53$} & $9.80 \pm 0.39$ & $16.48 \pm 0.35$ & $11.51 \pm 1.14$ & $19.88 \pm 1.48$ \\
    \bottomrule
  \end{tabular}
\end{table}

\begin{figure}[h!]
  \centering
  \newcommand{\trlpan}[2]{\begin{subfigure}[t]{0.32\linewidth}\centering\includegraphics[width=\linewidth]{#1}\caption*{\scriptsize #2}\end{subfigure}}
  \begin{minipage}[c]{0.87\textwidth}
    \trlpan{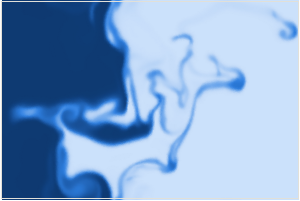}{Ground truth}\hfill
    \trlpan{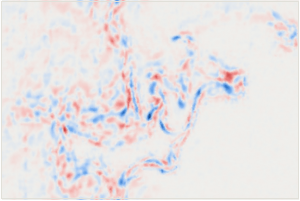}{Overtone+MNC (ours), 0.072}\hfill
    \trlpan{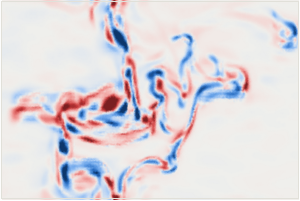}{GINO, 0.155}\\[2pt]
    \trlpan{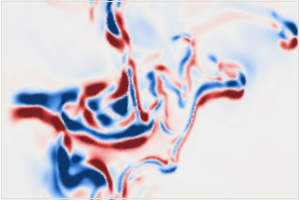}{LNO, 0.271}\hfill
    \trlpan{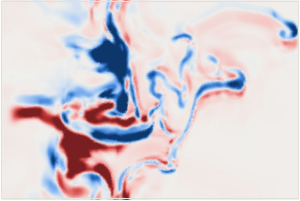}{Transolver++, 0.278}\hfill
    \trlpan{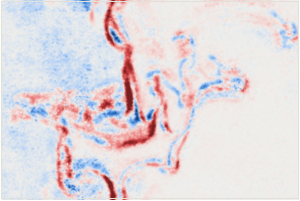}{MeshGraphNets, 0.148}
  \end{minipage}\hfill
  \begin{minipage}[c]{0.12\textwidth}\centering
    \includegraphics[height=2.0in]{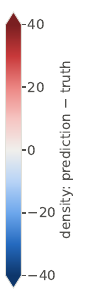}
  \end{minipage}
  \caption{Zero-shot one-step density error on the irregular TRL2D mesh.  Numbers are the window's density VRMSE; predicting no change (persistence) scores 0.308.}
  \label{fig:trl2d-zeroshot-onestep}
\end{figure}

\vspace{-0.25cm}

\subsection{Experiment 2: Fine-tuning on irregular meshes}\label{sec:finetune}

While zero-shot performance provides a useful baseline for judging the effectiveness of weight transfer onto the MNC tokeniser, it turns out that the mesh-native tokenisation also provides further advantages for fine-tuning models on unstructured meshes. In this experiment we compare the performance of the MNC retrofit against the baselines when fine-tuned on the irregular mesh, for 200 epochs (Appendix~\ref{app:protocol}) on each of the three Well datasets. The results for next-step VRMSE on the irregular TRL2D are shown in Fig.~\ref{fig:exp1-trl2d-finetune}, with full results for each dataset provided in Table~\ref{tab:exp1-finetuned}. The MNC retrofit leads to extremely fast fine-tuning: it reaches its plateau within ten epochs, while the same MNC-tokenised transformer trained from scratch needs about a hundred epochs to arrive at the same level. The grid-pretrained baselines start below their scratch counterparts but improve little (MeshGraphNets the most, levelling off well above the retrofit), indicating that their pretraining insights do not transfer as easily onto unstructured meshes. The MNC-tokenised transformer ends below every baseline whether or not it was pretrained, showing that the retrofit inherits a strong architecture as well as a strong initialisation.

Table~\ref{tab:exp1-finetuned} shows that this observation holds across all three datasets. After fine-tuning, the MNC retrofit achieves the lowest next-step and rollout VRMSE on TRL2D, active matter and shear flow, with errors a factor of two to three below the best baseline on active matter and shear flow. The benefit of grid pretraining is most pronounced on these two datasets, where it reduces the next-step error of the MNC model by about a quarter and a third, respectively, compared to training from scratch, while on TRL2D the pretrained and scratch models end up at similar levels.

\begin{figure}[h!]
  \centering
  \includegraphics[width=0.79\textwidth]{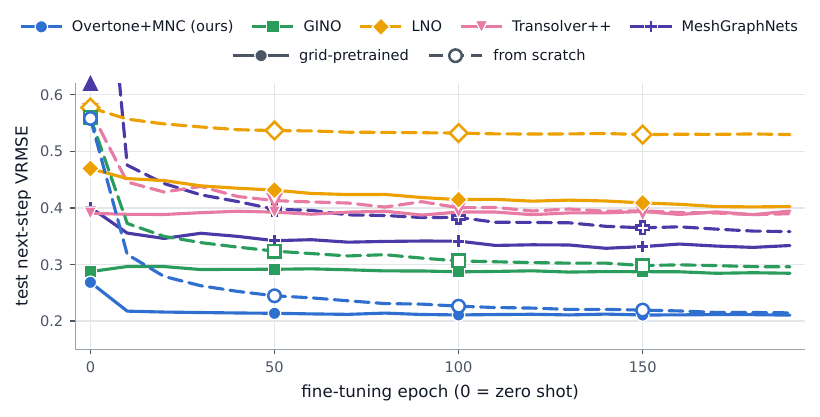}
  \caption{Test next-step VRMSE (mean) on the irregular TRL2D mesh against fine-tuning epoch. Grid-pretrained models are marked in solid lines and randomly initialised ones in dashed lines.}
  \label{fig:exp1-trl2d-finetune}
\end{figure}

\begin{table}[h!]
  \centering
  \caption{Fine-tuning on the irregular mesh: next-step and ten-step rollout VRMSE, in
  units of $10^{-2}$.}
  \label{tab:exp1-finetuned}
  \small
  \setlength{\tabcolsep}{2.5pt}
  \begin{tabular}{ll>{\columncolor{gray!12}}ccccc}
    \toprule
    & & Overtone+MNC & GINO & LNO & Transolver++ & MeshGraphNets \\
    \midrule
    \multicolumn{2}{l}{\emph{Next-step VRMSE}} &&&&&\\
    TRL2D         & pretrained & {\boldmath$22.10 \pm 0.04$} & $29.40 \pm 0.08$ & $40.89 \pm 0.12$ & $40.76 \pm 0.07$ & $33.91 \pm 0.43$ \\
                  & scratch    & $22.77 \pm 0.20$ & $29.78 \pm 0.27$ & $53.39 \pm 0.11$ & $39.78 \pm 0.13$ & $36.84 \pm 0.46$ \\
    \addlinespace
    Active matter & pretrained & {\boldmath$3.30 \pm 0.06$} & $10.71 \pm 0.04$ & $8.12 \pm 0.05$ & $12.81 \pm 0.04$ & $11.20 \pm 0.26$ \\
                  & scratch    & $4.36 \pm 0.08$ & $12.84 \pm 0.20$ & $15.86 \pm 1.27$ & $14.09 \pm 0.08$ & $13.35 \pm 0.03$ \\
    \addlinespace
    Shear flow    & pretrained & {\boldmath$0.43 \pm 0.01$} & $1.73 \pm 0.02$ & $3.56 \pm 0.12$ & $0.98 \pm 0.05$ & $1.78 \pm 0.05$ \\
                  & scratch    & $0.66 \pm 0.02$ & $2.12 \pm 0.01$ & $3.89 \pm 0.00$ & $1.34 \pm 0.05$ & $1.84 \pm 0.08$ \\
    \midrule
    \multicolumn{2}{l}{\emph{10-step rollout VRMSE}} &&&&&\\
    TRL2D         & pretrained & {\boldmath$49.04 \pm 1.61$} & $59.23 \pm 2.24$ & $68.29 \pm 1.38$ & $71.38 \pm 2.63$ & $75.11 \pm 2.39$ \\
                  & scratch    & $50.60 \pm 4.03$ & $64.14 \pm 1.62$ & $79.32 \pm 1.41$ & $69.63 \pm 1.08$ & $75.09 \pm 5.43$ \\
    \addlinespace
    Active matter & pretrained & {\boldmath$22.79 \pm 1.83$} & $63.71 \pm 1.81$ & $56.51 \pm 1.67$ & $91.08 \pm 2.33$ & $107.97 \pm 9.89$ \\
                  & scratch    & $29.73 \pm 1.62$ & $85.80 \pm 15.30$ & $81.59 \pm 2.52$ & $91.12 \pm 2.15$ & $102.29 \pm 11.67$ \\
    \addlinespace
    Shear flow    & pretrained & {\boldmath$2.65 \pm 0.16$} & $8.86 \pm 0.52$ & $17.32 \pm 3.18$ & $10.16 \pm 0.97$ & $30.27 \pm 9.40$ \\
                  & scratch    & $4.04 \pm 0.45$ & $12.05 \pm 2.36$ & $16.72 \pm 0.00$ & $13.07 \pm 1.21$ & $22.40 \pm 5.22$ \\
    \bottomrule
  \end{tabular}
\end{table}

\subsection{Experiment 3: Zero-shot retrofitting of a large-scale foundation model}\label{sec:multi-physics}
Our final experiment confirms that this retrofitting approach is not limited to small-scale models and datasets. Indeed, the MNC retrofit can provide an efficient bridge between large-scale grid-pretrained foundation models and unstructured mesh problems. In this experiment we retrofit Walrus~\citep{mccabe2026walrus}, a large-scale multi-physics foundation model pretrained on a multi-physics corpus, to the MeshGraphNets cylinder-flow benchmark. A particularly noteworthy aspect of this experiment is that the cylinder-flow dataset is not included in Walrus's pretraining corpus, so the model has no prior knowledge of this flow configuration. The MNC retrofit allows us to leverage the pretrained knowledge of Walrus and apply it to a new problem domain without requiring any additional training on the cylinder-flow dataset.
\looseness=-1
In Table~\ref{tab:walrus-cyl-vrmse} we report the next-step and ten-step rollout VRMSE for the retrofitted Walrus model, both in zero-shot and after fine-tuning on the cylinder-flow training split. We also include the performance of the baselines trained from scratch on the cylinder-flow dataset for comparison. In each of the fine-tuning/training configurations we allowed the same compute budget of 7.5h on a single NVIDIA GH200 GPU. The results show that (i) zero-shot, the retrofitted Walrus model outperforms all trained baselines except Transolver++, meaning that grid-pretrained knowledge is effectively transferred to the unstructured mesh problem even for unseen datasets, and (ii) the retrofitted Walrus model achieves the lowest next-step error after fine-tuning, outperforming all baselines by a significant margin.

\begin{table}[h!]
  \centering
  \caption{Cylinder flow dataset: next-step and ten-step rollout VRMSE, in units of $10^{-2}$. All training and fine-tuning runs were given the same compute budget of 7.5h on a single NVIDIA GH200 GPU.}
  \label{tab:walrus-cyl-vrmse}
  \small
  \begin{tabular}{lcc}
    \toprule
    & Next-step VRMSE & 10-step rollout VRMSE \\
    \midrule
    \rowcolor{gray!12} Walrus+MNC, pre-trained, fine-tuned  & {\boldmath$0.15 \pm 0.01$} & {\boldmath$1.02 \pm 0.00$} \\
    \rowcolor{gray!12} Walrus+MNC, pre-trained, zero-shot & $0.41$ & $3.30$ \\
    Walrus+MNC, random init, fine-tuned  & $0.44 \pm 0.01$ & $2.47 \pm 0.04$ \\
    \midrule
    MeshGraphNets, trained on cylinder flow  & $0.80 \pm 0.04$ & $8.26 \pm 0.03$ \\
    Transolver++, trained on cylinder flow  & $0.18 \pm 0.01$ & $2.05 \pm 0.16$ \\
    GINO, trained on cylinder flow  & $0.70 \pm 0.01$ & $3.80 \pm 0.21$ \\
    LNO, trained on cylinder flow  & $2.26 \pm 0.00$ & $12.10 \pm 0.00$ \\
    Persistence & $2.34$ & $12.50$ \\
    \bottomrule
  \end{tabular}
\end{table}\vspace{-0.15cm}

\looseness=-1
Figure~\ref{fig:cylinder-onestep} shows a one-step prediction of the $x$-velocity on a cylinder-flow test trajectory. Without any training on cylinder flow, the retrofitted Walrus predicts the vortex street with an error below all baseline models. Walrus was only pretrained on grid data and weights from its CNN tokeniser were loaded onto MNC. By contrast, each of the baselines was trained on the dataset with 7.5h of compute budget, demonstrating the advantage of transferring grid-based models onto unstructured datasets.
\begin{figure}[h!]
  \centering
  \newcommand{\cylpan}[2]{\begin{subfigure}[t]{0.49\linewidth}\centering\includegraphics[width=\linewidth]{#1}\caption*{\scriptsize #2}\end{subfigure}}
  \begin{minipage}[c]{0.87\textwidth}
    \cylpan{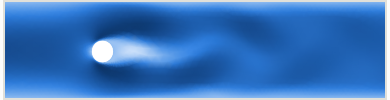}{Ground truth}\hfill
    \cylpan{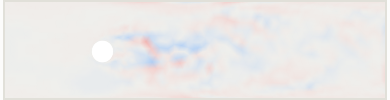}{Walrus+MNC, zero-shot, 0.0058}\\[2pt]
    \cylpan{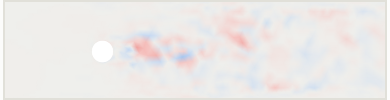}{Transolver++, 0.0052}\hfill
    \cylpan{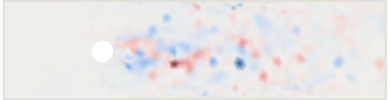}{GINO, 0.0096}\\[2pt]
    \cylpan{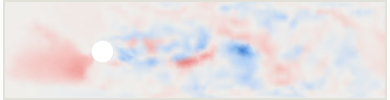}{MeshGraphNets, 0.0110}\hfill
    \cylpan{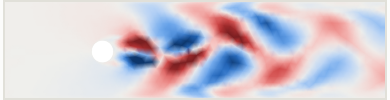}{LNO, 0.0314}
  \end{minipage}\hfill
  \begin{minipage}[c]{0.12\textwidth}\centering
    \includegraphics[height=2.0in]{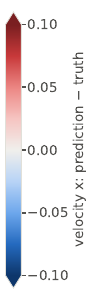}
  \end{minipage}
  \caption{Error in next-step $x$-velocity prediction on cylinder flow. Numbers are the window's $x$-velocity VRMSE; predicting no change (persistence) scores 0.0314.}
  \label{fig:cylinder-onestep}
\end{figure}%\vspace{-0.75cm}

\section{Discussion and conclusion}
\label{sec:discussion}

In this work we introduced mesh-native tokenisation as a means for retrofitting grid-pretrained physics emulators to unstructured meshes. The MNC layer is a simple locality-aware encoder that allows the transfer of pretrained weights from a CNN encoder/decoder pair to a mesh-native tokeniser, while retaining the source model's processor. We demonstrated that this approach allows for zero-shot transfer of grid-pretrained models onto irregular meshes, and that fine-tuning on the target mesh can be done efficiently and effectively. We also showed that this approach can be applied to large-scale foundation models, allowing for the transfer of knowledge from grid-based pretraining to unstructured mesh problems. This allows for significant computational savings through both repurposing of pretrained weights and reduced training time of processors by leveraging fast CNN-based pretraining~\citep{chetlur2014cudnnefficientprimitivesdeep}.

\textbf{Limitations and future work.} MNC retrofitting mainly improves sample efficiency, leading to quicker fine-tuning runs rather than overall improved accuracy after sufficient training. This suggests that the benefits of MNC are most pronounced in scenarios with limited training data or computational resources, and future work can focus on improving the performance of the retrofitted models in these scenarios, for example by introducing geometric quadrature weights for consistency under varying sampling density (Proposition~\ref{prop:continuum}). More broadly, since MNC acts on the meshes used by classical solvers, retrofitted emulators could serve as initial guesses for implicit time stepping~\citep{jin2024hybrid}, and symmetries of the underlying physics could be built into the kernel, as in equivariant reduced-order models~\citep{maierhofer2026essm}.

\subsection*{AI use statement}

In this work, we used generative AI tools to help formulate mathematical claims, to provide ingredients for and assist in writing the proofs of these claims, to provide feedback on the design of the experiments, to implement methods, and to clean and reformat datasets (in particular the conversion of grid and mesh data into a common format). We have not used generative AI tools to generate synthetic data sets, develop theoretical models or conceptual frameworks, propose or refine hypotheses, or interpret results, and translation and qualitative or thematic data analysis are not applicable to this work. Additionally, we used generative AI tools to create and edit software code, create and modify schematic figures, suggest experimental parameters, and draft parts of the paper. We have reviewed all AI-assisted work: all mathematical claims and proofs were checked line by line by the authors; AI-assisted code was reviewed and tested by the authors, including exactness checks of the weight transfer against the source models; all AI-drafted text was reviewed and edited by the authors, with every reported number checked against the underlying results; and all figures were checked against the underlying data. We take responsibility for the final content of this work, including text, claims or artifacts produced with the aid of generative AI.

\subsection*{Ethics statement}

This work uses only publicly available simulation data of physical systems (The Well and the MeshGraphNets cylinder-flow benchmark) and a publicly released pretrained model. No human subjects, personal data or other sensitive information were involved at any stage. All authors have read and adhere to the ICLR Code of Ethics, and we are not aware of any ethical concerns or potential for harm specific to this work beyond those of general-purpose scientific machine learning.

\subsection*{Reproducibility statement}

All data and pretrained weights used in this work are publicly available. The TRL2D, shear-flow and active-matter datasets are part of The Well~\citep{ohana2024well}, and the cylinder-flow benchmark is taken from the official MeshGraphNets~\citep{pfaff2021meshgraphnets} release. Appendix~\ref{app:datasets} provides details on the dataset splits, the construction of the irregular meshes and the resampling of the grid fields. The Walrus checkpoint retrofitted in Experiment~3 is the publicly released one~\citep{mccabe2026walrus}, used without modification of its weights. The MNC layer is specified in \S\ref{sec:method}, and complete proofs of Propositions~\ref{prop:exact} and~\ref{prop:continuum} and of the GCN counterexample are given in Appendix~\ref{app:proofs}. Appendix~\ref{app:models} lists every architecture with its reference implementation, configuration and parameter count, together with the training settings, and Appendix~\ref{app:protocol} gives the training and evaluation protocol, including the metrics, seeds and rollout procedure. After the review process we will release the complete code on GitHub, together with the trained model weights and the irregular meshes on Hugging Face.

\subsection*{Acknowledgements}

The authors gratefully acknowledge the use of resources provided by the Isambard-AI National AI Research Resource (AIRR) under project 0261-4743-5166-1 and 0261-5199-2243-1. Miles Cranmer acknowledges support from the Schmidt Sciences AI2050 Early Career Fellowship and the Isaac Newton Trust. Dr.
Mukhopadhyay thanks the Infosys-Cambridge AI center for support.

\bibliography{iclr2027_conference}
\bibliographystyle{iclr2027_conference}

\appendix
\include{appendix}

\end{document}

%% file: math_commands.tex
\usepackage{amsmath,amsfonts,bm}

\def\1{\bm{1}}

\DeclareMathAlphabet{\mathsfit}{\encodingdefault}{\sfdefault}{m}{sl}
\SetMathAlphabet{\mathsfit}{bold}{\encodingdefault}{\sfdefault}{bx}{n}

%% file: appendix.tex
\section{Transfer identities and consistency proofs}
\label{app:proofs}
\subsection{Proof of Proposition~\ref{prop:exact}}
\label{app:exact}
For completeness we repeat the statement of the Proposition:

\begingroup
\renewcommand{\theproposition}{\ref{prop:exact}}
\begin{proposition}[Exact grid reduction]
Suppose the input nodes $\{\bp_j\}$ are exactly the points of the source grid of spacing $(h_x,h_y)$, the physical kernel spacings equal the grid spacings, and the anchors and padding reproduce the output windows of a unit-dilation \texttt{Conv2d}. Load $\bK[\bu']=\Theta[\bu']$ for $\bu'\in\mathcal O$ and realise padding as zero-valued nodes with their occupancy included. Then \eqref{eq:enc} equals \eqref{eq:cnn} at every anchor node. Similarly, if $\widetilde\bK$ is loaded from the kernel weights of a \texttt{ConvTranspose2d} with unit dilation, stride equal to its kernel size and zero output padding, then \eqref{eq:dec} equals that \texttt{ConvTranspose2d} at every grid node.
\end{proposition}
\endgroup

\begin{proof}
The spline functions $\phi$ used in \eqref{eq:kernel} satisfy $\phi(m)=\mathbf1_{\{m=0\}}$ for integer values $m$. At each stencil offset $\bu\in\mathcal O$, therefore, $\bW(\bu)=\bK[\bu]$ and $\omega(\bu)=1$. Every other grid offset has zero kernel and occupancy $(\bW,\omega(\bu))=(0,0)$. A complete encoder window has $P=k_hk_w$ such offsets, so $s_{\hat v}=1$ and \eqref{eq:enc} is exactly \eqref{eq:cnn}. We note in the case of padding for the convolution, correct extension of the underlying mesh is required to ensure that the encoder window is complete. In particular, the encoder must be able to read the same number of samples as the source convolution window and, for zero padding nodes the features are zero but their occupancy still counts.

In particular, let $\Theta_{o,i,r,c}$ be a tensor stored by \texttt{Conv2d}, with $0\leq r<k_h$ and $0\leq c<k_w$. The offset loaded into the encoder is
\begin{equation}
\bu_{r,c}=\bigl(((k_w-1)/2-c)h_x,\;((k_h-1)/2-r)h_y\bigr),
\qquad \bK[\bu_{r,c}]_{o,i}=\Theta_{o,i,r,c}.
\label{eq:conv-indexing}
\end{equation}
Indeed, $\bp_{\hat v}-\bu_{r,c}$ is the input location read by the stored cross-correlation kernel. Even kernel sizes simply give half-integer offsets relative to the window centre and the differences between offsets remain integer multiples of $h_x,h_y$. Thus showing that \eqref{eq:enc} equals \eqref{eq:cnn} at every anchor node.

For \texttt{ConvTranspose2d}, a stored tensor $\Psi_{i,o,r,c}$ scatters from an anchor to the offset
\begin{equation}
\bv_{r,c}=\bigl((c-(k_w-1)/2)h_x,\;(r-(k_h-1)/2)h_y\bigr),
\qquad \widetilde{\bK}[\bv_{r,c}]_{o,i}=\Psi_{i,o,r,c}.
\label{eq:transpose-indexing}
\end{equation}
Here $\bv$ denotes a spatial offset. With unit dilation, stride equal to kernel size and zero output padding, every retained grid output belongs to exactly one discrete footprint. Other anchors can meet it only at zero spline values. Hence its occupancy is one, $t_j=1$, and \eqref{eq:dec} recovers the transposed convolution, including its bias. Output cropping and padding must match the source convention as before, and once this is guaranteed the proof shows that indeed \eqref{eq:dec} equals \texttt{ConvTranspose2d} at every grid node.
\end{proof}

\subsection{Multi-level hierarchies and the processor interface}
\label{app:multilevel}
For a hierarchy, apply the identities of Proposition~\ref{prop:exact} at each level with the corresponding coarse-grid spacing. Preserve the pointwise projections, activations, feature normalisation and token order. In \eqref{eq:rope}, the source token centres $p_a=p_a^{\min}+mH_a$, $m=0,\dots,n_a-1$, with $H_a$ the source token spacing, give $\mathrm{pos}_a=2m/(n_a-1)-1$, the $\mathrm{linspace}(-1,1,n_a)$ values of the source implementation, and hence exactly the source rotations (Figure~\ref{fig:rope}). The deterministic processor therefore receives and returns the same tensors. This proves equality of the complete grid models. Induction gives equality of autoregressive rollouts when the history update and normalisation are also identical. These are algebraic identities; finite-precision scatter reductions can produce small numerical discrepancies.

\begin{figure}[h!]
\centering
\includegraphics[width=\textwidth]{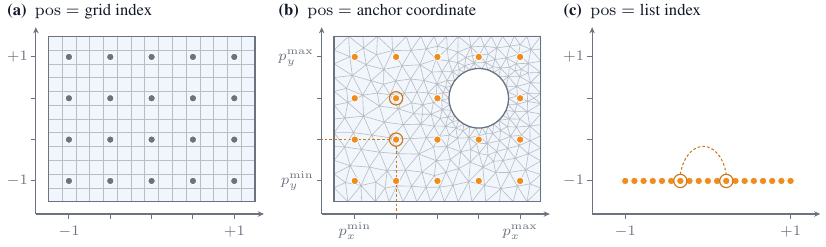}
\caption{\looseness=-1 Coordinate rotary position embedding for mesh tokens: only the source of $\mathrm{pos}$ changes. \textbf{(a)} On the grid, $\mathrm{pos}=\mathrm{linspace}(-1,1,\cdot)$ along each axis. \textbf{(b)} After retrofitting, $\mathrm{pos}$ is the anchor coordinate, rescaled by the outermost anchors~\eqref{eq:rope}. The window covered by the hole emits no token. On a lattice this recovers the angles of (a) exactly. \textbf{(c)} The stock axial embedding would place the flat list of anchors along one axis, separating vertical neighbours (ringed) and collapsing the second axis.}
\label{fig:rope}
\end{figure}

\newpage\subsection{Proof of Proposition~\ref{prop:continuum}}
\label{app:continuum}
For completeness we repeat the statement of the Proposition:

\begingroup
\renewcommand{\theproposition}{\ref{prop:continuum}}
\begin{proposition}[Continuum limit and discretisation invariance]
Suppose the mesh nodes satisfy $c_n\sum_j\delta_{\bp_j^{(n)}}\rightharpoonup\rho(y)\,dy$, with $c_n\downarrow0$ and $\rho\in L^1(\Omega)$ positive almost everywhere. For each fixed anchor whose full window lies inside $\Omega$, MNC with zero bias, $\bb=\mathbf0$, applied to $\bx_j=u(\bp_j^{(n)})$ has the joint limit
\begin{equation*}
\operatorname{MNC}_{h,n}u(\bp)\longrightarrow
|S|\,
\frac{\int_\Omega\kappa(\bp-y)u(y)\rho(y)\,dy}
{\int_{\bp-S}\rho(y)\,dy}.
\tag{\ref{eq:density-limit}}
\end{equation*}
Thus two equidistributed refinements of the same field, with the same limiting kernel, have asymptotically identical outputs at common anchors (Fig.~\ref{fig:limit}a).
\end{proposition}
\endgroup

Here $\rightharpoonup$ denotes weak convergence of measures, in the sense that
\[
c_n\sum_j \varphi\!\left(\bp_j^{(n)}\right)
\longrightarrow
\int_\Omega \varphi(y)\rho(y)\,dy
\qquad\text{as }n\to\infty
\]
for every bounded continuous test function $\varphi\colon\Omega\to\mathbb{R}$.

\begin{proof}
Recall from the setup in \S\ref{sec:theory} that $\Omega$ is bounded, $u\in C(\overline\Omega;\mathbb R^{F_{\rm in}})$, and $\kappa$ is a Lipschitz, matrix-valued function with $\operatorname{supp}\kappa\subseteq S=[-R_x,R_x]\times[-R_y,R_y]$; in particular it vanishes on $\partial S$. All vector norms are Euclidean and matrix norms are the corresponding induced norms. Keeping the \emph{entire} support fixed as the number of learned kernel weights grows means
\begin{equation}
h_x=\frac{2R_x}{k_w+1},\quad h_y=\frac{2R_y}{k_h+1},\quad
L_x=R_x-h_x,\quad L_y=R_y-h_y.
\label{eq:fixed-support}
\end{equation}
The zero ring then lies on $\partial S$ and $P_ha_h=k_hk_wh_xh_y\to|S|$.

Put $\kappa_h=\bW_h/a_h$, where the kernel weights are chosen as in \eqref{eq:continuum-kernel}. This is the bilinear interpolant of $\kappa$, extended by zero. Convexity of the interpolation weights gives
\begin{equation}
\|\kappa_h-\kappa\|_\infty\leq L_\kappa d_h,
\qquad d_h=(h_x^2+h_y^2)^{1/2},\quad L_\kappa=\Lip(\kappa),
\label{eq:interpolation-bound}
\end{equation}
since every point in a kernel cell is at most $d_h$ from its corners. Also
\begin{equation}
0\leq\omega_h\leq1,\qquad
\int_{\mathbb R^2}\omega_h(z)\,dz=P_ha_h,
\qquad
\mathbf1_{S^{-\varepsilon}}\leq\omega_h\leq\mathbf1_S
\label{eq:occupancy-limits}
\end{equation}
once $\max(h_x,h_y)<\varepsilon$, where $S^{-\varepsilon}$ is $S$ eroded by $\varepsilon$ on each side. Each tensor-product hat integrates to $a_h$, proving the middle identity.

Write $\mu_n=c_n\sum_j\delta_{\bp_j^{(n)}}$, with weak convergence understood against every function in $C(\overline\Omega)$. In particular, $\mu_n(\Omega)$ is bounded. Uniform convergence in \eqref{eq:interpolation-bound} and weak convergence give, jointly as $h_x,h_y\to0$ and $n\to\infty$,
\begin{equation}
\int\kappa_h(\bp-y)u(y)\,d\mu_n(y)
\longrightarrow \int_\Omega\kappa(\bp-y)u(y)\rho(y)\,dy.
\label{eq:numerator-limit}
\end{equation}
For the denominator, squeeze $\omega_h(\bp-\cdot)$ between the indicators of the translated rectangles in \eqref{eq:occupancy-limits}. Their boundaries have zero $\rho\,dy$ measure, so weak convergence applies to these indicators. Letting $\varepsilon\downarrow0$ yields
\[
\int\omega_h(\bp-y)\,d\mu_n(y)\longrightarrow
\int_{\bp-S}\rho(y)\,dy>0.
\]
Since $c_n\to0$, the unscaled occupancy sum tends to infinity, and the lower clipping in \eqref{eq:enc} is eventually inactive. The bias-free output can then be written as
\[
P_ha_h\,
\frac{\int\kappa_h(\bp-y)u(y)\,d\mu_n(y)}
{\int\omega_h(\bp-y)\,d\mu_n(y)}.
\]
Together with $P_ha_h\to|S|$, these limits prove \eqref{eq:density-limit}. For constant sampling density on the window, the density and window-area factors cancel. Applying this result to two sequences with the same field and kernel gives the stated asymptotic agreement.
\end{proof}

\subsection{A connectivity counterexample for a normalised GCN}
\label{app:gcn}
The following example shows that the discretisation invariance of Proposition~\ref{prop:continuum} need not hold for a mesh-edge GNN: on two triangulations of the same vertices, a degree-normalised GCN with a fixed positive scalar weight and no bias maps the same constant field to different values at the origin for every mesh width $\eta$, whereas MNC returns identical outputs.

Consider the scalar GCN layer of \citet{kipf2017semisupervised},
\begin{equation}
G_\Gamma(\bx)=\operatorname{ReLU}\!\left(w\,\widetilde D^{-1/2}(A_\Gamma+I)\widetilde D^{-1/2}\bx\right),
\qquad \widetilde D_{ii}=1+\deg_\Gamma(i),
\end{equation}
where $A_\Gamma$ is the adjacency matrix of the mesh graph $\Gamma$ (entries one for vertices sharing an edge, zero otherwise) and $w$ is the scalar weight. We take $w=1$ and no bias. Let the vertices be $\eta(i,j)$, $-n\leq i,j\leq n$, $\eta=1/n$, and $n\geq3$.

\begin{figure}[h!]
\centering
\includegraphics[width=0.84\textwidth]{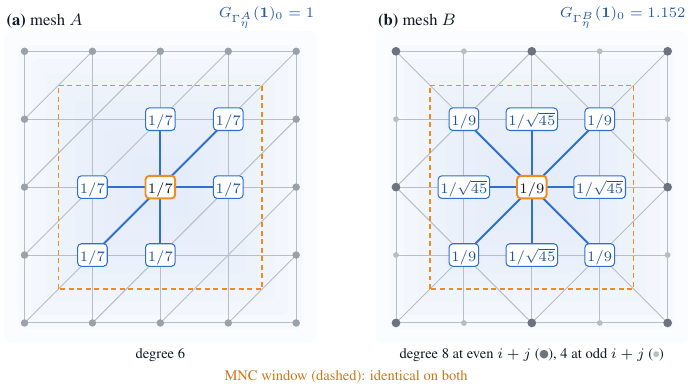}
\caption{Meshes $A$ and $B$ near the origin. Each vertex in the origin's GCN neighbourhood (blue) is labelled with its contribution to $G_\Gamma(\mathbf1)_0$, the origin's own label being its self-loop.}
\label{fig:gcn}
\end{figure}

In mesh $A$ (Figure~\ref{fig:gcn}a), split every square along its southwest-northeast diagonal. Every vertex in the origin's one-ring, including the origin, has degree six and self-loop degree seven. For $u\equiv1$, the seven contributions at the origin each equal $1/7$, giving $G_{\Gamma^A_\eta}(\mathbf1)_0=1$.

In mesh $B$ (Figure~\ref{fig:gcn}b), choose each square's diagonal to join its two vertices with even $i+j$. Even vertices have degree eight and odd vertices degree four. The origin has four axial odd neighbours and four diagonal even neighbours. Its self-loop and its four even neighbours each contribute $1/9$, while its four odd neighbours each contribute $1/\sqrt{9\cdot5}$. Therefore
\[
G_{\Gamma^B_\eta}(\mathbf1)_0=\frac59+\frac4{\sqrt{45}}
=\frac59+\frac4{3\sqrt5}=1.1518403495\ldots.
\]
Both values are positive, so ReLU leaves them unchanged. Both triangulations comprise right isosceles triangles of diameter $\sqrt2\eta$, and both represent the same constant field exactly. The discrepancy also persists for any positive scalar weight, which simply multiplies both outputs.

With common anchors, physical kernel spacings and coefficients, MNC has the same geometric neighbours and occupancies on these two meshes and hence exactly the same output. This counterexample isolates connectivity, since the sample locations themselves do not change, and shows that for the degree-normalised GCN the output depends on how the vertices are connected, not only on the sampled field, even as $\eta\to0$.

\section{Datasets}
\label{app:datasets}
This appendix provides further details on the datasets and meshes used in the numerical experiments. In all cases, every sample is a window of six consecutive context snapshots and the following target snapshot. Figures~\ref{fig:ds-tlr2d}-\ref{fig:ds-cyl} show each dataset's meshes and representative fields; Table~\ref{tab:datasets} lists the sizes. The forecasting experiments use the official TRL2D and active-matter splits and a trajectory-level resplit of the shear-flow training files (Appendix~\ref{app:protocol}); the reconstruction experiments use the official splits throughout (Appendix~\ref{app:autoencoding}).

\textbf{Turbulent radiative layer (TRL2D).} A two-dimensional mixing layer between a cold dense and a hot dilute gas with radiative cooling, on a $128\times384$ grid ($x\in[-0.5,0.5]$ periodic, $y\in[-1,2]$ open). Fields: density, pressure and two velocity components; density spans two orders of magnitude and is displayed in $\log_{10}$. Nine values of the cooling time $t_\mathrm{cool}\in\{0.03,\dots,3.16\}$ with 8/1/1 trajectories each for train/validation/test, 101 snapshots per trajectory ($\Delta t=1.6$). The mesh twin has $49{,}152$ nodes and $146{,}433$ edges. The pressure field is nearly spatially uniform (spatial standard deviation $\approx2\%$ of its magnitude), which makes its VRMSE ill-conditioned; the reconstruction experiments therefore report normalised MSE.

\begin{figure}[h!]
\centering
\begin{subfigure}[b]{0.30\textwidth}
\centering\includegraphics[width=\linewidth]{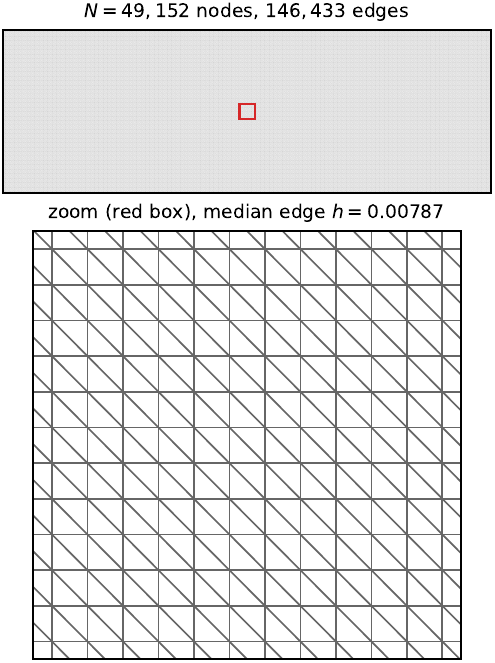}
\caption{mesh and zoom}\label{fig:ds-tlr2d-mesh}
\end{subfigure}\hfill
\begin{subfigure}[b]{0.33\textwidth}
\centering\includegraphics[width=\linewidth]{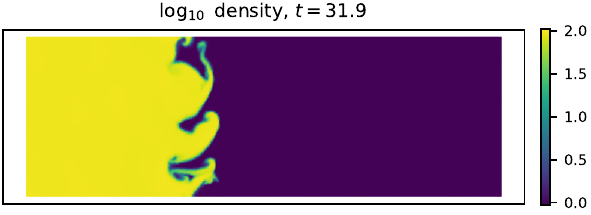}
\caption{$\log_{10}\rho$, $t=31.9$}\label{fig:ds-tlr2d-f1}
\end{subfigure}\hfill
\begin{subfigure}[b]{0.33\textwidth}
\centering\includegraphics[width=\linewidth]{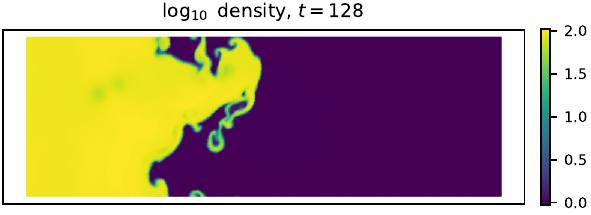}
\caption{$\log_{10}\rho$, $t=128$}\label{fig:ds-tlr2d-f2}
\end{subfigure}
\caption{TRL2D: the regular mesh twin with the zoom region marked in red and the zoomed mesh below it (left), and $\log_{10}$ density at two times along a $t_\mathrm{cool}=0.18$ trajectory. The Well mesh twins are triangulations of the original grids, so grid and mesh tokenisers see exactly the same samples.}
\label{fig:ds-tlr2d}
\end{figure}

\begin{figure}[h!]
\centering
\begin{subfigure}[b]{0.30\textwidth}
\centering\includegraphics[width=\linewidth]{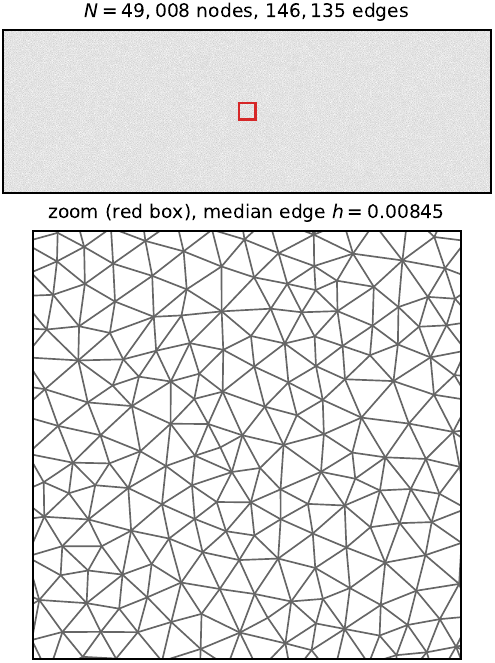}
\caption{mesh and zoom}\label{fig:ds-tlr2d-irr-mesh}
\end{subfigure}\hfill
\begin{subfigure}[b]{0.33\textwidth}
\centering\includegraphics[width=\linewidth]{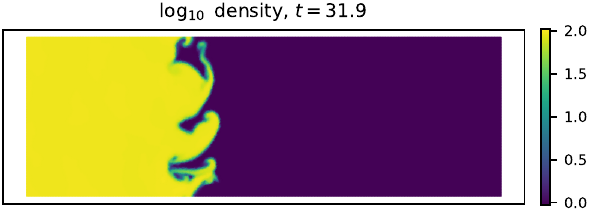}
\caption{$\log_{10}\rho$, $t=31.9$}\label{fig:ds-tlr2d-irr-f1}
\end{subfigure}\hfill
\begin{subfigure}[b]{0.33\textwidth}
\centering\includegraphics[width=\linewidth]{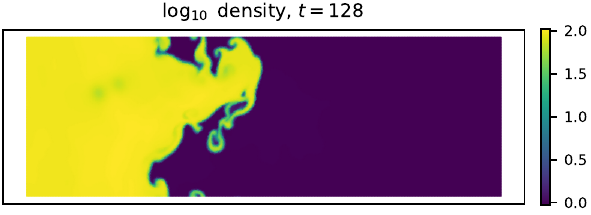}
\caption{$\log_{10}\rho$, $t=128$}\label{fig:ds-tlr2d-irr-f2}
\end{subfigure}
\caption{TRL2D on the Experiment~1 irregular mesh: the Gmsh Delaunay triangulation ($49{,}008$ nodes, $146{,}135$ edges, node count matched to the grid) with the zoom region marked in red (left), and the same $t_\mathrm{cool}=0.18$ density snapshots as Figure~\ref{fig:ds-tlr2d}, bilinearly resampled onto the mesh nodes which are exactly the samples the zero-shot-evaluated models see.}
\label{fig:ds-tlr2d-irr}
\end{figure}

\newpage
\textbf{Shear flow.} Two-dimensional incompressible shear layers with a passive tracer on a doubly periodic $256\times512$ grid over $[0,1]^2$, for 28 combinations of Reynolds number ($10^4$ to $5\cdot10^5$) and Schmidt number ($0.1$ to $10$) with 40 initial conditions each, 200 snapshots per trajectory ($\Delta t=0.1$). Fields: tracer, pressure and velocity. Depending on the initial interface thickness a trajectory either rolls up into Kelvin--Helmholtz billows (Figure~\ref{fig:ds-shear}) or stays close to its laminar profile, so the test sets mix both regimes. The mesh twin has $131{,}072$ nodes and $391{,}681$ edges. The reconstruction transfer experiments additionally triangulate the domain at matched resolution as a uniform triangular lattice and as an irregular Delaunay mesh (Figure~\ref{fig:auto-meshes}, Appendix~\ref{app:auto-transfer}); the forecasting experiments use the irregular mesh of Figure~\ref{fig:ds-shear-irr}.

\begin{figure}[h!]
\centering
\begin{subfigure}[b]{0.30\textwidth}
\centering\includegraphics[width=\linewidth]{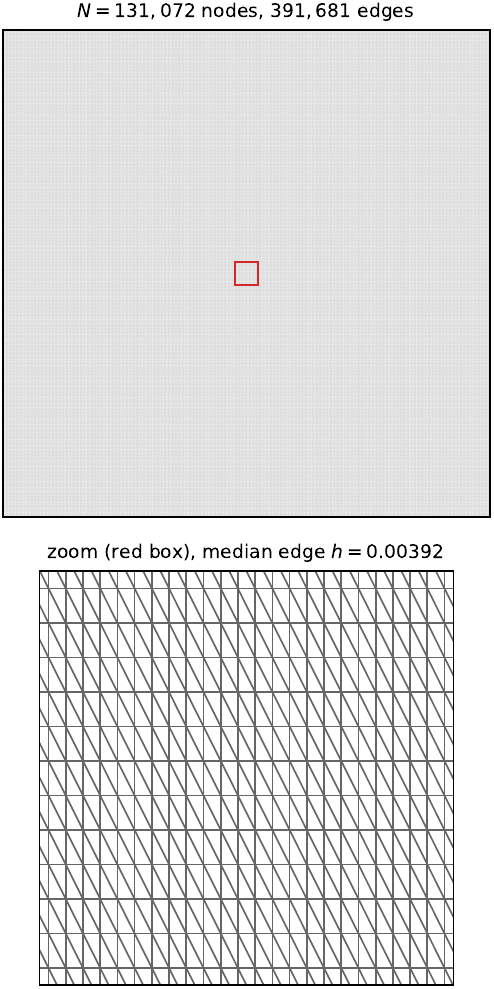}
\caption{mesh and zoom}\label{fig:ds-shear-mesh}
\end{subfigure}\hfill
\begin{subfigure}[b]{0.33\textwidth}
\centering\includegraphics[width=\linewidth]{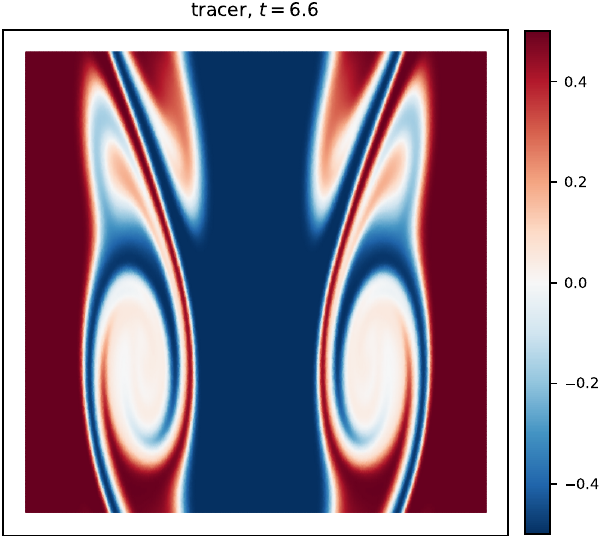}
\caption{tracer, $t=6.6$}\label{fig:ds-shear-f1}
\end{subfigure}\hfill
\begin{subfigure}[b]{0.33\textwidth}
\centering\includegraphics[width=\linewidth]{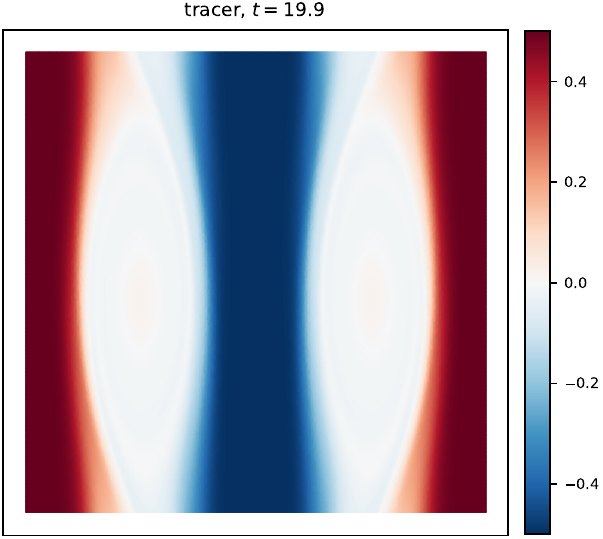}
\caption{tracer, $t=19.9$}\label{fig:ds-shear-f2}
\end{subfigure}
\caption{Shear flow: the mesh twin with zoom (left), and the tracer at two times along a $Re=10^4$, $Sc=1$ trajectory that develops Kelvin--Helmholtz billows.}
\label{fig:ds-shear}
\end{figure}

\begin{figure}[h!]
\centering
\begin{subfigure}[b]{0.30\textwidth}
\centering\includegraphics[width=\linewidth]{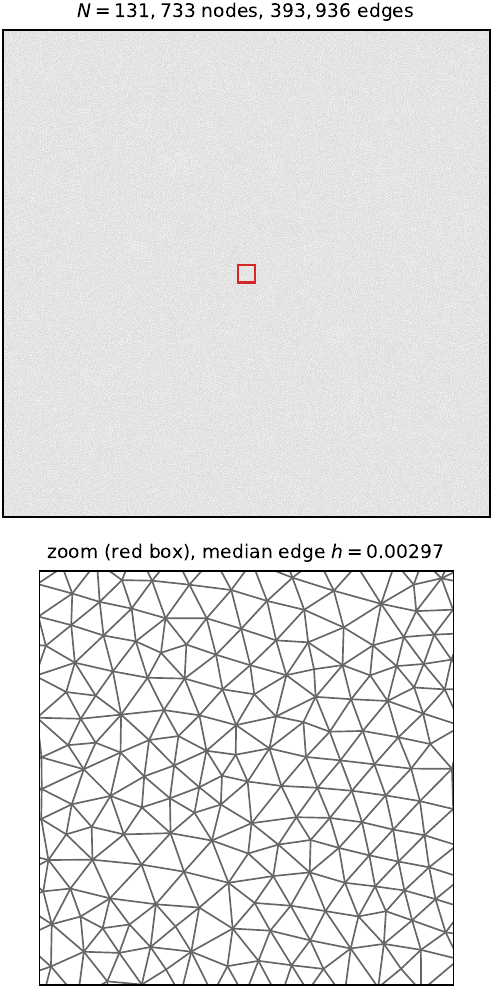}
\caption{mesh and zoom}\label{fig:ds-shear-irr-mesh}
\end{subfigure}\hfill
\begin{subfigure}[b]{0.33\textwidth}
\centering\includegraphics[width=\linewidth]{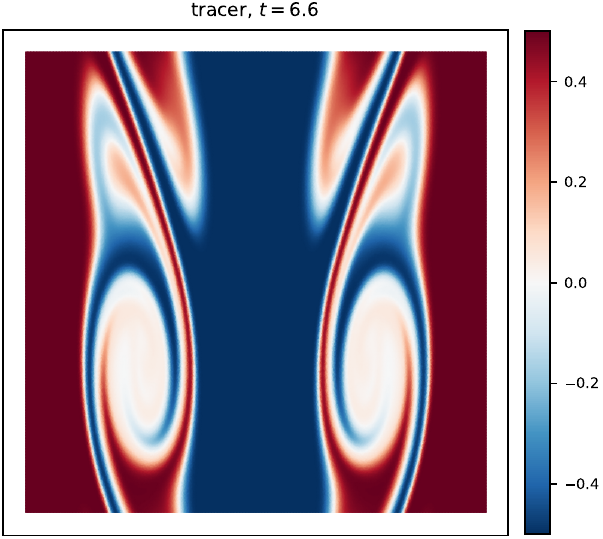}
\caption{tracer, $t=6.6$}\label{fig:ds-shear-irr-f1}
\end{subfigure}\hfill
\begin{subfigure}[b]{0.33\textwidth}
\centering\includegraphics[width=\linewidth]{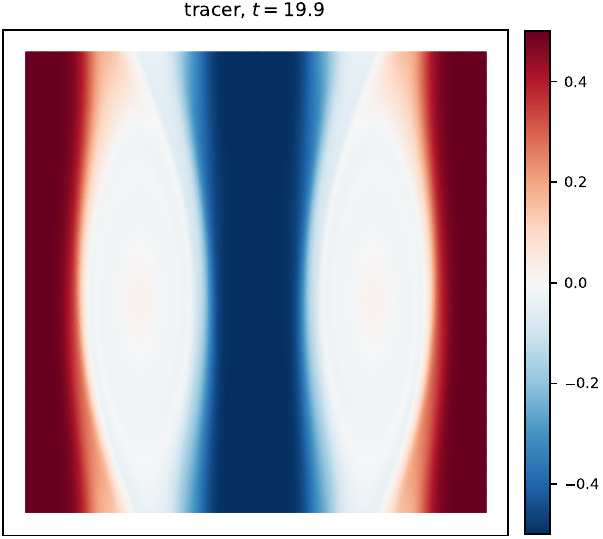}
\caption{tracer, $t=19.9$}\label{fig:ds-shear-irr-f2}
\end{subfigure}
\caption{Shear flow on the Experiment~1 irregular mesh ($131{,}733$ nodes, $393{,}936$ edges, node count matched to the grid): the same $Re=10^4$, $Sc=1$ tracer snapshots as Figure~\ref{fig:ds-shear}, bilinearly resampled onto the mesh nodes (trajectory from the official validation file; Experiment~1 itself uses the trajectory-level resplit of the training files).}
\label{fig:ds-shear-irr}
\end{figure}

\textbf{Active matter.} A two-dimensional continuum model of active rod-like particles in a Stokes fluid, on a doubly periodic $256\times256$ grid over $[0,10]^2$. Fields: the concentration, the two velocity components, and the orientation tensor $D$ and strain-rate tensor $E$ with four components each, eleven channels in total; all of them, including the tensor fields, are inputs and targets. The sweep covers the active dipole strength $\alpha\in\{-1,\dots,-5\}$ and the alignment parameter $\zeta\in\{1,3,\dots,17\}$ at fixed $\beta=0.8$, 45 settings with five trajectories each that differ in the plane-wave perturbation of the initially isotropic orientation tensor; each trajectory has 81 snapshots over $t\in[0,20]$ ($\Delta t=0.25$).
We use the official split (175/24/26 trajectories), and the source is trained at the native resolution, giving $16\times16=256$ tokens. The on-grid target is the lattice of the $65{,}536$ grid points. The irregular mesh of Experiment~1 is a Gmsh Delaunay triangulation of the square at a single characteristic length, with the node count matched to the grid: $65{,}632$ nodes and $130{,}374$ triangles, with edge lengths between $0.61h$ and $1.76h$ (median $1.07h$) in units of the grid spacing $h$. Periodicity is not built into the mesh; the grid fields are resampled onto its nodes. Unlike TRL2D and shear flow, no lattice twin is stored for active matter, so Figure~\ref{fig:ds-am-irr} shows the irregular mesh only.

\begin{figure}[h!]
\centering
\begin{subfigure}[b]{0.30\textwidth}
\centering\includegraphics[width=\linewidth]{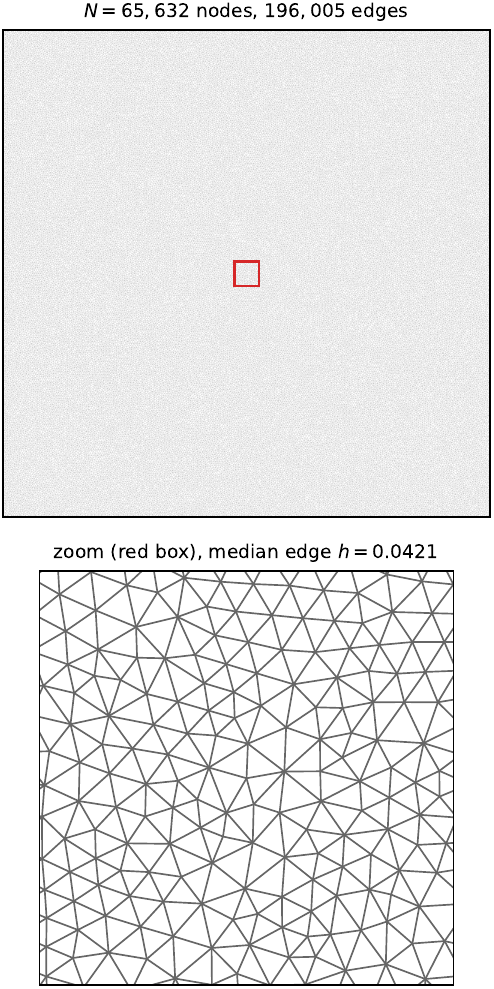}
\caption{mesh and zoom}\label{fig:ds-am-irr-mesh}
\end{subfigure}\hfill
\begin{subfigure}[b]{0.33\textwidth}
\centering\includegraphics[width=\linewidth]{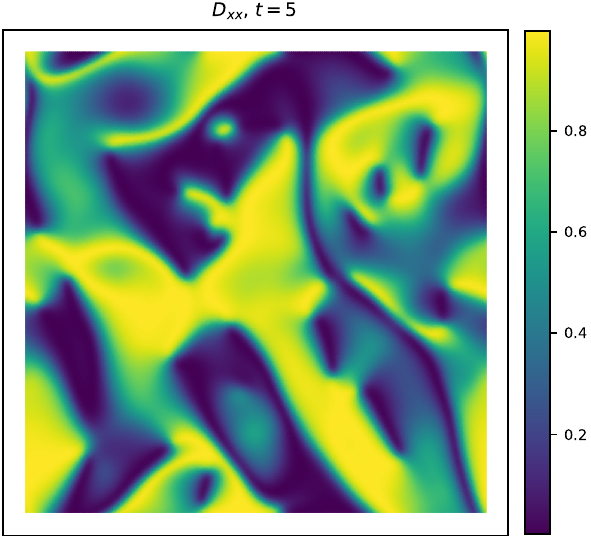}
\caption{$D_{xx}$, $t=5$}\label{fig:ds-am-irr-f1}
\end{subfigure}\hfill
\begin{subfigure}[b]{0.33\textwidth}
\centering\includegraphics[width=\linewidth]{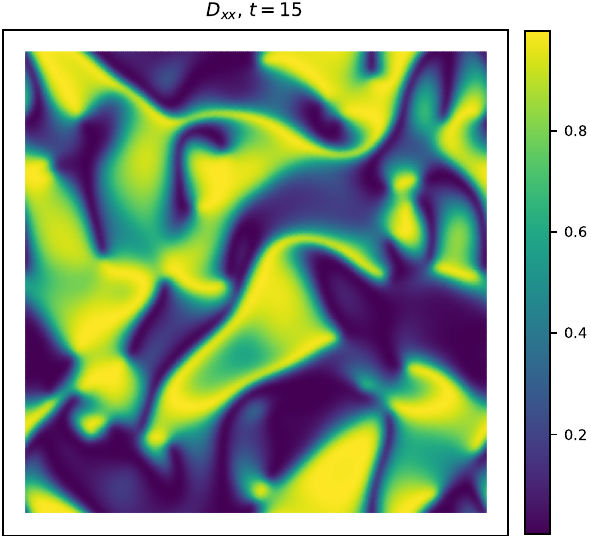}
\caption{$D_{xx}$, $t=15$}\label{fig:ds-am-irr-f2}
\end{subfigure}
\caption{Active matter on the Experiment~1 irregular mesh ($65{,}632$ nodes, $196{,}005$ edges, node count matched to the grid): the mesh with the zoom region marked in red (left), and the orientation-tensor component $D_{xx}$ at two times along a $\zeta=9$, $\alpha=-3$ test trajectory, bilinearly resampled onto the mesh nodes. The concentration is nearly uniform ($1\pm5\cdot10^{-4}$), so a tensor channel is shown.}
\label{fig:ds-am-irr}
\end{figure}
\newpage
\textbf{Cylinder flow.} The MeshGraphNets incompressible flow past a cylinder~\citep{pfaff2021meshgraphnets} (COMSOL), 1,000/100/100 trajectories for train/validation/test, each on \emph{its own} unstructured triangulated mesh ($1{,}732$-$2{,}059$ nodes, median $1{,}885$) of the $1.6\times0.41$ channel, refined towards the cylinder and the walls. Fields: pressure and velocity; 600 snapshots per trajectory ($\Delta t=0.01$). Node types mark inflow, outflow and wall nodes. This is the only dataset without a grid representation; the MNC hierarchy is built per trajectory. It is used in Experiment~3, on refined meshes (Appendix~\ref{app:cylinder-refinement}), and in the reconstruction experiments on the native meshes (Appendix~\ref{app:autoencoding}).

\begin{figure}[h!]
\centering
\begin{subfigure}[b]{0.30\textwidth}
\centering\includegraphics[width=\linewidth]{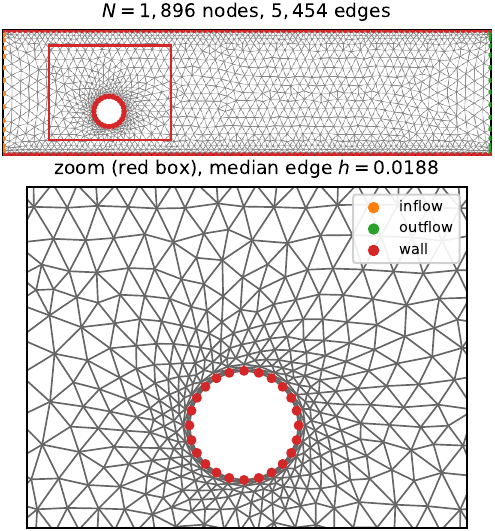}
\caption{mesh and zoom}\label{fig:ds-cyl-mesh}
\end{subfigure}\hfill
\begin{subfigure}[b]{0.33\textwidth}
\centering\includegraphics[width=\linewidth]{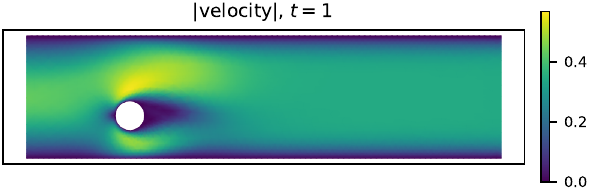}
\caption{$|\mathbf{v}|$, $t=1$}\label{fig:ds-cyl-f1}
\end{subfigure}\hfill
\begin{subfigure}[b]{0.33\textwidth}
\centering\includegraphics[width=\linewidth]{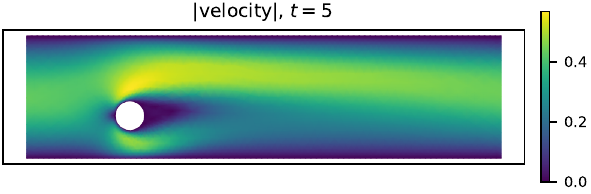}
\caption{$|\mathbf{v}|$, $t=5$}\label{fig:ds-cyl-f2}
\end{subfigure}
\caption{Cylinder flow: one of the per-trajectory meshes, with inflow, outflow and wall nodes marked in the zoom (left), and velocity magnitude at two times.}
\label{fig:ds-cyl}
\end{figure}

\subsection{Mesh refinement for cylinder flow}
\label{app:cylinder-refinement}
The native cylinder-flow meshes are far coarser than the grids on which Walrus was pretrained, so that a pretrained token window would contain only a handful of nodes. For Experiment~3 each mesh is therefore refined uniformly three times. Every refinement splits each triangle into four at its edge midpoints, and the value at a new midpoint is the mean of the values at the two end points of its edge, so the piecewise-linear field of the original mesh is reproduced exactly; no new information is added. New boundary midpoints inherit the inflow, outflow or wall type of their edge, and the cylinder remains the polygon of the original mesh. The median node count of the training meshes grows from $1{,}885$ to $7{,}306$, $28{,}756$ and $114{,}088$ over the three levels ($104{,}184$-$125{,}224$ nodes at level three). Splits are unchanged and normalisation statistics are recomputed per level. The first nodes of every refined mesh are the original nodes, in their original order and with their original values, so all predictions are scored on the original mesh nodes (Appendix~\ref{app:protocol}).

\begin{table}[h!]
\centering
\caption{Datasets summary. Trajectories are given as train/validation/test of the official splits; nodes and edges refer to the mesh representation (undirected edges): the lattice twin for TRL2D and shear flow, and the Experiment~1 irregular mesh for active matter, which has no lattice twin.}
\label{tab:datasets}
\footnotesize
\setlength{\tabcolsep}{2.8pt}
\resizebox{\textwidth}{!}{%
\begin{tabular}{lcccccc}
\toprule
Dataset & Domain / resolution & Fields & Trajectories & Steps & Nodes & Edges \\
\midrule
TRL2D & $128\times384$ grid, $[-0.5,0.5]\times[-1,2]$ & $\rho,\,p,\,v_x,\,v_y$ & 72/9/9 & 101 & $49{,}152$ & $146{,}433$ \\
Shear flow & $256\times512$ grid, $[0,1]^2$ periodic & tracer, $p,\,v_x,\,v_y$ & 896/112/112 & 200 & $131{,}072$ & $391{,}681$ \\
Active matter & $256\times256$ grid, $[0,10]^2$ periodic & $c,\,v_x,\,v_y,\,D,\,E$ (11 ch.) & 175/24/26 & 81 & $65{,}632$ & $196{,}005$ \\
Cylinder flow & unstructured $1.6\times0.41$ channel & $p,\,v_x,\,v_y$ & 1000/100/100 & 600 & $1{,}732$-$2{,}059$ & $\approx 5{,}400$ \\
\bottomrule
\end{tabular}}
\end{table}

\section{Model details}
\label{app:models}
Every model in \S\ref{sec:forecasting} is a complete autoregressive surrogate with its own processor. The baselines are run from their official or reference implementations, with minimal changes to incorporate the current data format, and are configured as in their publications wherever the data allow; the departures are listed below. Table~\ref{tab:models} summarises the architectures and Table~\ref{tab:training} the training settings. Parameter counts are for the complete model and were obtained by instantiating each configuration; they differ between datasets only through the number of input and output channels.

\begin{table}[h!]
\centering
\caption{Architectures of Experiments~1 and~2. Parameter counts are for the four-channel datasets (TRL2D, shear flow) and, in brackets, for the eleven-channel active-matter dataset. MNC and the grid source have identical counts, since the transplant replaces operations, not parameters.}
\label{tab:models}
\small
\setlength{\tabcolsep}{4pt}
\renewcommand{\arraystretch}{1.15}
\begin{tabular}{@{}>{\raggedright\arraybackslash}p{2.3cm}>{\raggedright\arraybackslash}p{2.9cm}>{\raggedright\arraybackslash}p{5.75cm}>{\raggedleft\arraybackslash}p{1.95cm}@{}}
\toprule
Model & Implementation & Configuration & Parameters \\
\midrule
MNC retrofit\newline \citep{mukhopadhyay2026overtone} & Overtone-family\newline patch modulator & 12 space-time blocks, width 384, 6 heads;\newline patch 16 (two $4\times4$, stride-4 stages) & 26,301,955\newline (same) \\
\addlinespace
GINO\newline \citep{li2023gino} & \texttt{neuraloperator} 2.0.0,\newline \texttt{GINO\_Poisson2d}\newline reference & FNO width 64, 4 layers, $20\times20$ modes;\newline latent lattice of 64 cells on the short axis;\newline GNO radius 1.5 cells & 5,099,036\newline (5,122,381) \\
\addlinespace
LNO\newline \citep{wang2024lno} & official code & width 640, 8 blocks,\newline 512 latent modes, 8 heads & 29,901,316\newline (29,932,683) \\
\addlinespace
Transolver++\newline \citep{luo2025transolver++} & official code & 8 layers, width 256, 8 heads,\newline 32 slices, MLP ratio 1 & 2,303,308\newline (2,326,611) \\
\addlinespace
MeshGraphNets\newline \citep{pfaff2021meshgraphnets} & PyTorch\newline re-implementation & 15 message-passing steps, latent 128;\newline 3-hidden-layer MLPs with LayerNorm & 2,879,748\newline (2,886,027) \\
\bottomrule
\end{tabular}
\end{table}

\textbf{MNC retrofit (Experiments~1 and~2).} The source is a patch-modulated forecasting transformer of the Overtone family~\citep{mukhopadhyay2026overtone}: a convolutional encoder of two $4\times4$, stride-four stages (channel widths $96\to384$, GroupNorm with 12 groups) lifts each frame to tokens of an overall patch size of 16, and a processor of 12 space-time blocks, each a full spatial attention with axial rotary position embedding followed by a causal axial attention in time, with query-key LayerNorm and a SwiGLU feed-forward of ratio four, evolves them; a mirrored transposed-convolution decoder returns field increments. Input and output use a fixed vocabulary of field embeddings, so the parameter count does not depend on the dataset. This gives 192 tokens on TRL2D, 256 on active matter and 512 on shear flow. The retrofit loads every tensor of the grid source by name. The four tokeniser kernels are converted into MNC kernel weights (Proposition~\ref{prop:exact}): the two encoder stages by the convolution indexing \eqref{eq:conv-indexing} and the two decoder stages, with the bias of the last one, by the transposed indexing \eqref{eq:transpose-indexing}. The conversion aborts if any tensor is neither copied nor converted. The mesh hierarchy inherits the source's kernel spacings and window origin and must reproduce the source's token count. The rotary embedding of the spatial attention is evaluated at the anchor coordinates \eqref{eq:rope}. Before any retrofit is used, a check on the source grid's own points confirms that the transplanted model reproduces the source to a relative error below $10^{-3}$.

\textbf{GINO.} We use \texttt{neuralop.models.GINO} from \texttt{neuraloperator} 2.0.0 in the library's \texttt{GINO\_Poisson2d} reference configuration: input GNO MLP $[256,512,256]$, output GNO MLP $[512,1024,512]$, a 16-channel positional embedding, linear GNO transforms with \texttt{half\_cos} weighting, and an FNO with lifting and projection ratio four and group normalisation. The latent lattice has 64 square cells along the short axis of the domain ($64\times192$ on TRL2D, $64\times64$ on active matter), and the GNO radius is set to 1.5 latent cells instead of the reference value, which would produce an impractical number of edges on grids of $5\cdot10^4$-$1.3\cdot10^5$ nodes.

\textbf{LNO.} We use the official Latent Neural Operator code with vanilla attention, GELU activations and no time input, sized to the same parameter class as the MNC source (width 640, 8 blocks, 512 latent modes, 8 heads) rather than the smaller configurations of~\citet{wang2024lno}. Positions are normalised by the domain's $x$-extent.

\textbf{Transolver++.} We use the official Transolver++ code without modification except that its multi-GPU \texttt{all\_reduce} calls are removed for single-GPU runs. The official repository releases only the million-scale model, so the size follows the configuration of the paper's standard benchmarks (8 layers, width 256, 8 heads, 32 slices, MLP ratio 1). As in the official code, the Gumbel noise on the slice weights is also drawn at evaluation, and node coordinates are concatenated to the input features.

\textbf{MeshGraphNets.} We use an established PyTorch re-implementation with the published processor size (15 message-passing steps, latent width 128). Its MLPs have one hidden layer more than DeepMind's (2.9M instead of 2.3M parameters) and use Xavier initialisation; the decoder is widened to the dataset's channel count. Edge features are the relative position and its length, standardised once per mesh. On the grid the model runs on a triangulated lattice with a fixed diagonal. To keep the comparison with the other models controlled, it receives the same per-sample-normalised stacked frames as they do, without node-type one-hot encoding, running normalisers or training noise.

\textbf{Common interface of the baselines.} All four baselines read the six context frames stacked as input channels, with the node coordinates, and predict the increment of the last frame; the MNC model is trained on the increments at all six causal positions. All models share the per-sample, area-weighted normalisation of fields and increments and the area-weighted MSE loss (Appendix~\ref{app:protocol}).

\begin{table}[h!]
\centering
\caption{Training settings of Experiments~1 and~2, shared by all architectures. One ``epoch'' in Figure~\ref{fig:exp1-trl2d-finetune} is 100 mesh-training updates.}
\label{tab:training}
\small
\setlength{\tabcolsep}{4pt}
\begin{tabular}{ll}
\toprule
Setting & Value \\
\midrule
Optimiser & AdamW, $\beta=(0.9,0.95)$, weight decay $10^{-4}$, gradient-norm clip 5 \\
Learning rate & $10^{-4}$ for every architecture and stage \\
Grid pretraining & 50,000 updates at batch 16; 500 warm-up updates, cosine decay to $5\%$ \\
Mesh training (retrofit and scratch) & 20,000 updates at batch 8; 100 warm-up updates, constant rate \\
Seeds & three per mesh-training arm; retrofits start from the seed-0 source \\
Context / target & six frames / field increment \\
Regularisation & drop path 0.1 (MNC only); no dropout \\
Precision / hardware & fp32 (TF32 disabled); one NVIDIA GH200 (Isambard-AI) per run \\
\bottomrule
\end{tabular}
\end{table}

Large baselines are trained with gradient accumulation over equal micro-batches, so the effective batch is unchanged. By design, one learning rate is used for every architecture and stage; we did not tune it per architecture.

\textbf{Walrus retrofit (Experiment~3).} We use the released Walrus checkpoint~\citep{mccabe2026walrus} (Hugging Face \texttt{polymathic-ai/walrus}, 1.3B parameters) with the unmodified upstream code. Walrus is an isotropic space-time transformer: 40 blocks of width 1408, each a full spatial attention with 16 heads and axial rotary embedding followed by an axial time attention, and a two-stage adaptive-stride convolutional encoder and decoder whose first stage selects the embedding columns of the fields present. Cylinder flow is presented in the format of the FlowBench flow-past-object data in Walrus's pretraining set: pressure and velocity (padded with a zero third component) on a $512\times128$ canvas over the $1.6\times0.41$ channel, with the cylinder as a constant mask field, open inflow and outflow boundaries, and walls at the top and bottom. MNC replaces only the outermost convolution on each side: the first encoder stage gathers directly from the mesh nodes and the last decoder stage scatters back to them; the inner stages, normalisations and processor are unchanged, and the processor keeps Walrus's axial rotary embedding on the token grid. The MNC kernel weights are Walrus's own kernels, so the retrofit adds no parameters. Normalisation is Walrus's samplewise RMS, area-weighted over the channel, and the model predicts the next increment from six frames. The meshes are uniformly refined three times, reproducing the piecewise-linear field exactly (median $1.1\cdot10^5$ nodes), and all scores are computed on the original mesh nodes (Appendix~\ref{app:cylinder-refinement}).

Fine-tuning trains all parameters with AdamW (weight decay $10^{-4}$), learning rate $10^{-5}$ after 20 warm-up updates, gradient-norm clip 10 and batch eight (one window per micro-batch, with activation recomputation), in fp32 with TF32 matrix products on one GH200, for a 7.5-hour budget (about 1.4k updates). The random-initialisation arm uses the same architecture and recipe from Walrus's own initialisation. The cylinder baselines use the implementations and sizes of Table~\ref{tab:models}, trained from scratch at mesh level three on the same windows, normalisation, optimiser and time budget, with learning rates $10^{-3}$ and $10^{-4}$ tried and the better one reported ($10^{-3}$ for MeshGraphNets, Transolver++ and GINO; $10^{-4}$ for LNO). Being non-causal, they are trained on the last increment only.

\section{Forecasting protocol}
\label{app:protocol}
This appendix records how the autoregressive experiments of \S\ref{sec:forecasting} are trained and evaluated. The architectures and optimiser settings are given in Appendix~\ref{app:models}; the tokeniser reconstruction comparison in Appendix~\ref{sec:tokeniser-experiments} uses a different protocol, detailed there. All experiments were performed on NVIDIA GH200 GPUs with 96\,GB of memory, one GPU per run.

\subsection{Splits and windows}
We use the official train/validation/test splits for TRL2D and active matter (Table~\ref{tab:datasets}). For shear flow we resplit the official training files by trajectory, stratified by $(Re,Sc)$, into 727/84/84 trajectories after dropping one duplicate; grid pretraining and all mesh training use only the training trajectories of this split. A sample is a window of six consecutive context frames and the following target frame. Every start time with stride one defines a window, giving 95, 75 and 194 windows per trajectory for TRL2D, active matter and shear flow. Training passes over all windows in a fresh random order each pass, seeded by the run seed. The test windows are drawn once with a fixed seed that does not depend on the run, so every model of a dataset is evaluated on the same windows. Intermediate evaluations use a fixed subset of 128 test windows. The final evaluation uses all 855 TRL2D test windows and 1,024 windows each for active matter and shear flow; these contain the 128.

\subsection{Normalisation and training objective}
Each window is normalised with its own statistics. Per channel, the mean and standard deviation of the six context frames and, separately, of their five consecutive increments are computed with the quadrature weights $w_j$ and averaged over time. The model receives the normalised context and predicts normalised increments, which are mapped back to physical fields by adding the denormalised increment to the last context frame. The loss is the area-weighted mean-squared error of the normalised increments. For the MNC model it is taken at all six causal positions; the baselines, which are not causal in time, predict only the increment from the last context frame to the target.

Spatial means use quadrature weights $w_j$, normalised to sum to one: the cell area on the grid and, on the irregular mesh, the lumped P1 mass $\tfrac13\sum_{T\ni j}|T|$ of node $j$. They enter the statistics, the loss and the metrics, so that sums over nodes approximate integrals over the domain whatever the node density. They are not inserted into the MNC aggregation, whose normalisation remains \eqref{eq:enc}; this distinction matters for interpreting Proposition~\ref{prop:continuum}.

\subsection{Meshes and field resampling}
The irregular mesh of each dataset is a Gmsh Delaunay triangulation of the grid's bounding rectangle at one characteristic length, without physics-based grading, with boundary nodes on the rectangle's edges and the node count matched to the grid (within $0.5\%$ for all three datasets; Appendix~\ref{app:datasets}). Field values at its nodes are bilinear interpolants of the simulation arrays, computed from cached weights. The targets are therefore samples of the same grid data, not solutions recomputed on a different discretisation.

The physical window origin and kernel spacings of the MNC hierarchy are inherited from the source grid and held fixed, and source and target hierarchies have the same token counts: 192 for TRL2D, 256 for active matter and 512 for shear flow. Windows are not wrapped across periodic boundaries: at the edge of the domain they contain only the samples inside it. Neighbour lists, relative offsets, spline weights and occupancies depend only on the geometry and are cached.

\subsection{Metrics}
\label{app:metrics}
All metrics are computed on denormalised physical fields. For a predicted frame $\hat u$ with truth $u$ and a channel $c$, the variance-normalised RMSE~\citep{ohana2024well} is
\begin{equation}
\mathrm{VRMSE}_c(\hat u,u)=\Bigl(\frac{\sum_j w_j\,(\hat u_{j,c}-u_{j,c})^2}{\sum_j w_j\,(u_{j,c}-\bar u_c)^2+10^{-6}}\Bigr)^{1/2},\qquad \bar u_c=\sum_j w_j\,u_{j,c},
\label{eq:vrmse}
\end{equation}
that is, the relative $L^2$ error with respect to the deviation of the truth from its spatial mean. On the grid all $w_j$ are equal and this is the VRMSE of The Well. Scores are averaged over all $C$ channels of a dataset, including the tensor fields of active matter: $\mathrm{VRMSE}(\hat u,u)=\frac1C\sum_{c=1}^C\mathrm{VRMSE}_c(\hat u,u)$.

\textbf{Next-step VRMSE.} For a test window $w$ with context ending at frame $t_w$, the model predicts $\hat u_{t_w+1}$ from the true context, and the window score is $e_w=\mathrm{VRMSE}(\hat u_{t_w+1},u_{t_w+1})$.

\textbf{Rollout VRMSE.} A rollout starts from a true six-frame context ending at frame $t_s$ and predicts $K=10$ frames autoregressively: each prediction $\hat u^{(s)}_{t_s+k}$ is appended to the context, the oldest frame is dropped and the normalisation statistics are recomputed before the next call, so no ground truth enters after the initial context. The rollout score is the channel-mean VRMSE averaged over the ten steps,
\begin{equation}
r_s=\frac1K\sum_{k=1}^{K}\mathrm{VRMSE}\bigl(\hat u^{(s)}_{t_s+k},\,u_{t_s+k}\bigr).
\label{eq:rollout-vrmse}
\end{equation}
Rollouts start at five evenly spaced times in every test trajectory.

\textbf{Aggregation.} For the tables, window and rollout scores are first averaged within each test trajectory $\tau$, giving $\bar e_\tau$ over its test windows and $\bar r_\tau$ over its five starts, and the reported next-step and rollout VRMSE are the medians of $\bar e_\tau$ and $\bar r_\tau$ over test trajectories. The training curves of Figure~\ref{fig:exp1-trl2d-finetune} use the same aggregation, restricted to the 128 fixed test windows at each evaluation. Transolver++ draws Gumbel noise in its slice weights at evaluation, so its scores are single draws.

\subsection{Evaluation schedule and reported checkpoints}
Sources are evaluated every 5,000 of their 50,000 updates and mesh-training runs every 1,000 of their 20,000 updates (21 points); rollouts are run at 0, 25k and 50k source updates and at 0, 5k, 10k and 20k mesh updates. All reported numbers are from the final checkpoint of a run; the validation split is monitored but not used for model or checkpoint selection. The zero-shot scores of Experiment~1 are the evaluation at step zero of the retrofit run, that is, the grid-pretrained source transplanted onto the irregular mesh before any mesh update, evaluated on the full set of test windows and rollouts. For MNC the weights are loaded directly from the grid-pretrained CNN components. The spread in Table~\ref{tab:exp1-zeroshot} is over grid sources pretrained with three seeds. The mesh-training runs of Experiment~2 all start from the seed-0 source, and their spread is over three mesh-training seeds.

\subsection{Experiment 3: cylinder flow}
Experiment~3 is scored on the original MeshGraphNets meshes, although the models run on meshes refined three times: the first nodes of each refined mesh are the original ones, in order and with their original values, so all levels are compared on the same nodes. Next-step VRMSE uses 16 random start times in each of the first ten test trajectories (160 windows), and ten-step rollouts use 16 other start times per trajectory. Here VRMSE weights all nodes equally: per channel (pressure and the two velocity components) it is the RMSE over the original nodes divided by the spatial standard deviation of the truth, averaged over channels, with the median over windows reported.

\section{Additional numerical results}
\label{app:additional}
Figure~\ref{fig:trl2d-onestep-grid} shows the one-step density error of every model on a single TRL2D test window on the irregular mesh (test trajectory 7, frame 79 predicted from frames 73-78), for the four starting points of Experiments~1 and~2: the grid-pretrained model zero-shot and after fine-tuning on the mesh, and the same architecture trained on the mesh from scratch, at initialisation and after training. The retrofitted Overtone+MNC model already resolves the mixing layer zero-shot, with a fainter error field than any baseline at any stage, and fine-tuning reduces its error further; trained from scratch, the same architecture ends slightly above the fine-tuned retrofit. Predicting no change (persistence) scores $0.308$ on this window.

\begin{figure}[h!]
  \centering\scriptsize
  \setlength{\tabcolsep}{1pt}
  \newcommand{\trlimg}[1]{\includegraphics[width=0.228\textwidth]{images/campaign/trl2d_onestep_panels/#1}}
  \begin{tabular}{cccc}
    \trlimg{traj7_t73_truth.pdf} & \trlimg{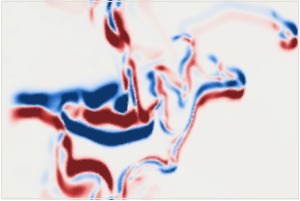} & \trlimg{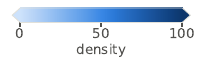} & \trlimg{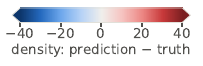} \\
    Ground truth & Persistence (no change), 0.308 & & \\[3pt]
    Pretrained, zero-shot & Pretrained, fine-tuned & Scratch, initialisation & Scratch, trained \\
    \trlimg{traj7_t73_mnc_zeroshot_err.pdf} & \trlimg{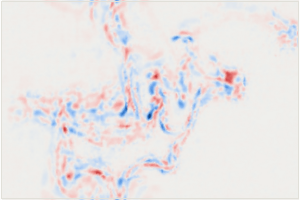} & \trlimg{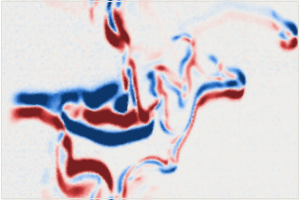} & \trlimg{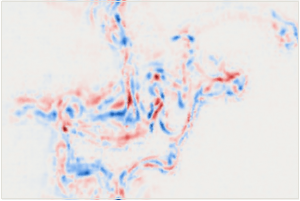} \\
    Overtone+MNC (ours), 0.072 & 0.066 & 0.309 & 0.086 \\[2pt]
    \trlimg{traj7_t73_gino_zeroshot_err.pdf} & \trlimg{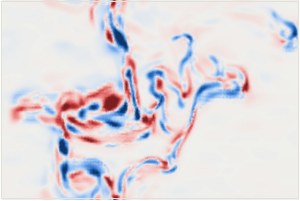} & \trlimg{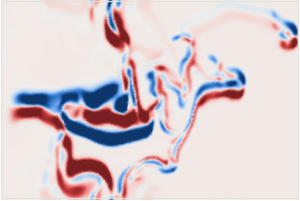} & \trlimg{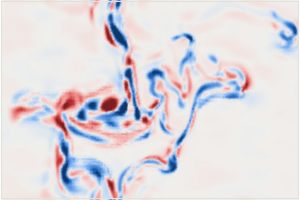} \\
    GINO, 0.155 & 0.154 & 0.309 & 0.157 \\[2pt]
    \trlimg{traj7_t73_lno_zeroshot_err.pdf} & \trlimg{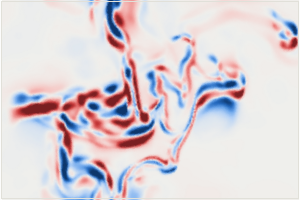} & \trlimg{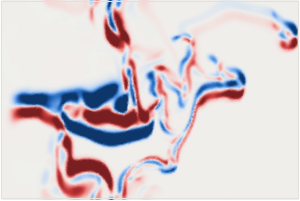} & \trlimg{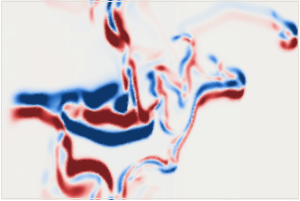} \\
    LNO, 0.271 & 0.205 & 0.308 & 0.308 \\[2pt]
    \trlimg{traj7_t73_transolverpp_zeroshot_err.pdf} & \trlimg{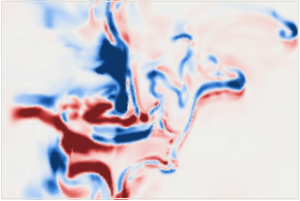} & \trlimg{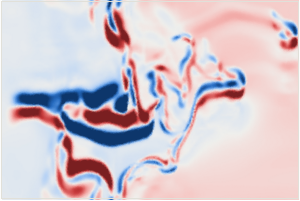} & \trlimg{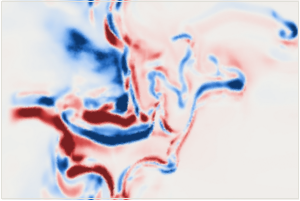} \\
    Transolver++, 0.278 & 0.270 & 0.321 & 0.235 \\[2pt]
    \trlimg{traj7_t73_mgn_zeroshot_err.pdf} & & \trlimg{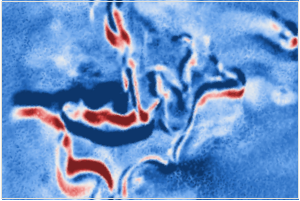} & \trlimg{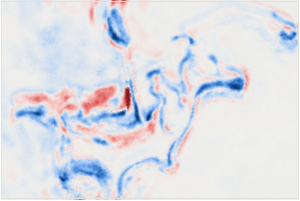} \\
    MeshGraphNets, 0.148 & & 0.624 & 0.116 \\[2pt]
  \end{tabular}
  \caption{One-step density error (prediction $-$ truth) on the irregular TRL2D mesh, test trajectory 7, frame 79: grid-pretrained models zero-shot and after fine-tuning, and the same architectures trained from scratch, at initialisation and after 20k updates. Numbers are the density VRMSE of this window. The fine-tuned MeshGraphNets panel was not exported for this window.}
  \label{fig:trl2d-onestep-grid}
\end{figure}

\section{Tokeniser reconstruction and transfer}
\label{sec:tokeniser-experiments}
\label{app:autoencoding}
In this appendix we isolate the encoder-decoder through snapshot reconstruction, comparing MNC with the tokenisers of GINO~\citep{li2023gino} and LNO~\citep{wang2024lno}, without their latent processors. The complete GINO and LNO models of the forecasting experiments are much larger and have their own processors (Appendix~\ref{app:models}); their results should not be conflated with those reported here. CNN and MNC have about 17k parameters, GINO about 22k, and LNO about 21k with 32 latent modes. CNN, MNC and GINO have matched bottleneck sizes; LNO's smaller bottleneck is examined separately in Appendix~\ref{app:lno-parity}.

\begin{table}[h!]
  \centering
  \caption{Tokeniser reconstruction from scratch after 100 epochs: test normalised MSE, mean $\pm$ standard deviation over three seeds. The Well grid twins use identical samples; cylinder flow uses native meshes. LNO uses 32 modes; latent-parity results are in Appendix~\ref{app:lno-parity}.}
  \label{tab:autoencoding}
  \small
  \begin{tabular}{lccc}
    \toprule
    Tokeniser & TRL2D & Shear flow & Cylinder flow \\
    \midrule
    CNN (grid) & {\boldmath$\pmsci{1.62}{0.16}{-3}$} & $\pmsci{1.18}{0.03}{-3}$ & --- \\
    MNC (ours) & $\pmsci{1.68}{0.06}{-3}$ & {\boldmath$\pmsci{1.15}{0.08}{-3}$} & {\boldmath$\pmsci{2.58}{0.01}{-2}$} \\
    GINO & $\pmsci{1.79}{0.01}{-2}$ & $\pmsci{5.25}{0.01}{-2}$ & $\pmsci{2.29}{0.10}{-1}$ \\
    LNO (32 modes) & $\pmsci{9.93}{0.02}{-2}$ & $\pmsci{7.37}{0.74}{-2}$ & $\pmsci{9.12}{0.64}{-2}$ \\
    \bottomrule
  \end{tabular}
\end{table}

Table~\ref{tab:autoencoding} compares training from random initialisation on two Well datasets~\citep{ohana2024well} and native cylinder-flow meshes~\citep{pfaff2021meshgraphnets}. On the Well grid twins all models see identical samples: MNC matches the CNN, while the other tokenisers have higher errors under this protocol. On cylinder flow MNC has the lowest reconstruction error, without requiring a grid representation.

\begin{figure}[h!]
\centering
\begin{minipage}[c]{0.47\textwidth}
\textbf{Grid-to-mesh transfer.} We pretrain CNN, GINO and LNO on the $256\times512$ shear-flow grid for 100 epochs. CNN weights are loaded into MNC; GINO and LNO load their own pretrained weights, retaining the target geometry. Each is then fine-tuned on the full training split, bilinearly resampled onto an irregular triangulation at matched resolution. Dashed curves use the same target architectures, data and mesh-training budgets, but start from random weights.
\end{minipage}\hfill
\begin{minipage}[c]{0.51\textwidth}
\centering\includegraphics[width=\linewidth]{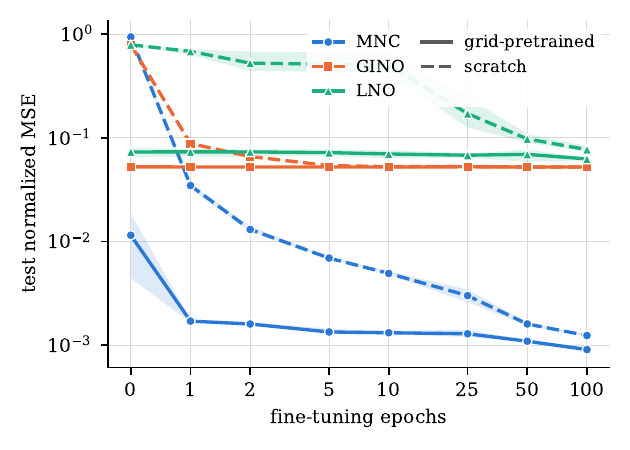}
\end{minipage}
\caption{Shear-flow autoencoding after a change of discretisation. Solid curves: CNN$\to$MNC, GINO$\to$GINO and LNO$\to$LNO transfer. Dashed curves: the corresponding models trained from scratch. Bands show variation across three seeds. Epoch zero is before any mesh training; source pretraining is additional to the horizontal-axis budget.}
\label{fig:transfer-shear}
\end{figure}

MNC starts near $10^{-2}$ normalised MSE and reaches about $1.5\times10^{-3}$ after one mesh-training epoch, close to its scratch-trained error after 100 epochs (Figure~\ref{fig:transfer-shear}). GINO and LNO also benefit from pretraining at initialisation, but remain at substantially higher errors and approach their scratch-trained levels as training proceeds. Thus the result is not that other tokenisers cannot transfer: the CNN$\to$MNC construction preserves a more accurate field representation and makes it available after far less mesh training. Figure~\ref{fig:overview}c shows the corresponding zero-shot reconstructions. This compares the usefulness of the trained tokenisers, not remeshing penalties at matched source accuracy. Lattice and cross-domain controls are in Appendices~\ref{app:auto-transfer} and~\ref{app:cylinder-transfer}.

\FloatBarrier
\subsection{Data, architecture and training}
\textbf{CNN and MNC.} A pointwise projection maps the input to 16 channels, followed by a $4\times4$/stride-four layer and a $2\times2$/stride-two layer to 64 channels, the mirrored decoder, and a pointwise output projection. GroupNorm and GELU are applied between stages. There are no skip connections, and only the output projection has a bias. Parameter counts are 16,805 for the five-channel interface (the four Well fields plus a zero channel) and 16,739 for the three-channel cylinder example. The overall reduction factor is eight, rather than 16 in forecasting.

\textbf{Protocol.} Only the latest frame is reconstructed; six context frames supply samplewise, per-channel RMS normalisation. Training uses MSE in these normalised variables, Adam at $10^{-3}$, and 100 epochs of 200 sampled windows at batch size one, without warm-up or a learning-rate schedule. All three tokeniser families use this protocol for source pretraining and mesh fine-tuning. Reconstruction MSE is not the physical-field VRMSE of \S\ref{sec:forecasting}.

\textbf{Grid twins and native meshes.} The Well controls use identical field samples for CNN, MNC and the baselines. TRL2D has $128\times384$ samples; shear flow has $256\times512$. The shear reconstruction experiments retain the original dataset split, not the forecasting resplit. The MeshGraphNets cylinder-flow dataset~\citep{pfaff2021meshgraphnets} has a separate native mesh per trajectory, approximately 1.9k nodes, and three fields: pressure and two velocities (Figure~\ref{fig:ds-cyl}, Appendix~\ref{app:datasets}). We use its 1,000/100/100 training/validation/test split and evaluate 500 sampled test windows. MNC builds a hierarchy for each geometry; in this dataset its two levels typically contain about 132 and 33 anchors. This tests reconstruction on a domain containing an obstacle, not grid-pretrained cylinder forecasting.

\subsection{GINO and LNO tokeniser configurations}
\label{app:tokeniser-baselines}
\textbf{GINO.} We use the \texttt{neuraloperator} implementation with its FNO processor removed (\texttt{fno\_n\_layers=0}): input GNO, pointwise lifting, output GNO and projection. Kernel MLPs act on relative coordinates with a 16-channel sinusoidal embedding; the latent width is 64. There are 21,674 parameters with five fields and 21,350 with three. The regular latent grid matches the CNN/MNC bottleneck: $16\times48$ for TRL2D, $32\times64$ for shear flow and $3\times11$ for cylinder flow, giving 49,152, 131,072 and 2,112 latent scalars. The input radius is five mesh cells and the output radius one latent cell. These are broader, overlapping windows, rather than exactly the CNN's effective footprint. Kernel weights do not depend on the latent-grid size, so the same model can be evaluated with target geometry after pretraining.

\textbf{LNO.} We use the cross-attention encoder and decoder with the latent processor removed (\texttt{n\_block=0}). Position and field MLPs produce 48-channel embeddings. Position-dependent scores are normalised over input nodes to encode each latent mode and over modes to decode at each query point. The main comparison and all transfer experiments use 32 modes, or 1,536 latent scalars, with approximately 21k parameters. Unlike GINO and MNC, this is a smaller bottleneck rather than a latent-size-matched control. Increasing the mode count also changes the parameters and cost; Table~\ref{tab:lno-parity} reports the available matched-bottleneck ablation.

\label{app:lno-parity}
\begin{table}[h!]
  \centering
  \caption{LNO mode-count ablation: test normalised MSE after 100 epochs of direct
  training (run means). The 32-mode and MNC rows repeat Table~\ref{tab:autoencoding} for
  reference.}
  \label{tab:lno-parity}
  \small
  \begin{tabular}{lcc}
    \toprule
    Tokeniser & TRL2D & Shear flow \\
    \midrule
    LNO (32 modes) & $\sci{9.93}{-2}$ & $\sci{7.37}{-2}$ \\
    LNO (latent parity) & $\sci{9.84}{-2}$ & $\sci{4.17}{-2}$ \\
    MNC (ours) & {\boldmath$\sci{1.68}{-3}$} & {\boldmath$\sci{1.15}{-3}$} \\
    \bottomrule
  \end{tabular}
\end{table}

The latent-parity variants use 1,024 modes on TRL2D and 2,730 on shear flow (49,152 and 131,040 latent scalars, respectively). They have 69,669 and 153,263 parameters. The larger bottleneck improves shear reconstruction, but does not close the gap to MNC under the reported protocol. This is a direct-training control, not an additional latent-parity transfer experiment.

\subsection{Execution cost}
\begin{table}[h!]
  \centering
  \caption{Mean training time per epoch in seconds for the reconstruction study (each model's summed
  per-step time on the same single NVIDIA GH200; mean $\pm$ standard deviation over three seeds).}
  \label{tab:timings}
  \small
  \begin{tabular}{lccc}
    \toprule
    Tokeniser & TRL2D & Shear flow & Cylinder flow \\
    \midrule
    CNN (grid) & \textbf{0.55 $\pm$ 0.04} & \textbf{0.61 $\pm$ 0.03} & --- \\
    MNC (ours) & 3.28 $\pm$ 0.19 & 3.46 $\pm$ 0.10 & 5.80 $\pm$ 0.18 \\
    GINO & 1.67 $\pm$ 0.02 & 3.89 $\pm$ 0.09 & 1.32 $\pm$ 0.05 \\
    LNO (32 modes) & 1.85 $\pm$ 0.00 & 4.45 $\pm$ 0.01 & \textbf{0.81 $\pm$ 0.02} \\
    \bottomrule
  \end{tabular}
\end{table}
The timings in Table~\ref{tab:timings} count training steps on one NVIDIA GH200 and exclude one-off geometry construction. MNC is approximately six times slower than dense CNN execution on these grid twins, and it is not uniformly faster than the other mesh tokenisers. The argument for transfer is the preservation of a useful representation and a reduction in target-mesh training, not intrinsically cheaper sparse layers. These measurements are not timings of the larger forecasting models.

\FloatBarrier
\subsection{Grid-to-mesh transfer and companion lattice control}
\label{app:auto-transfer}
For Figure~\ref{fig:transfer-shear}, MNC loads the CNN's projections, convolutional and transposed-convolutional weights, GroupNorm affine parameters and output bias. No MNC parameters are newly initialised in this within-dataset transplant. The source's physical kernel spacings and window origin are retained. GINO and LNO replace only the geometry buffers with the target mesh geometry.

The irregular target and the uniform triangular lattice in Figure~\ref{fig:auto-meshes} have matched resolution, but their points do not coincide with the source-grid samples. Field values are obtained by bilinear resampling. Figure~\ref{fig:transfer-lattice} repeats the comparison on the uniform lattice: the relative accuracy and fast initial recovery of transferred MNC persist. Neither comparison is an exactness test on identical inputs, because the sample positions have changed.

\begin{figure}[h!]
\centering
\begin{subfigure}[b]{0.47\textwidth}
\centering\includegraphics[width=\linewidth,height=0.21\textheight,keepaspectratio]{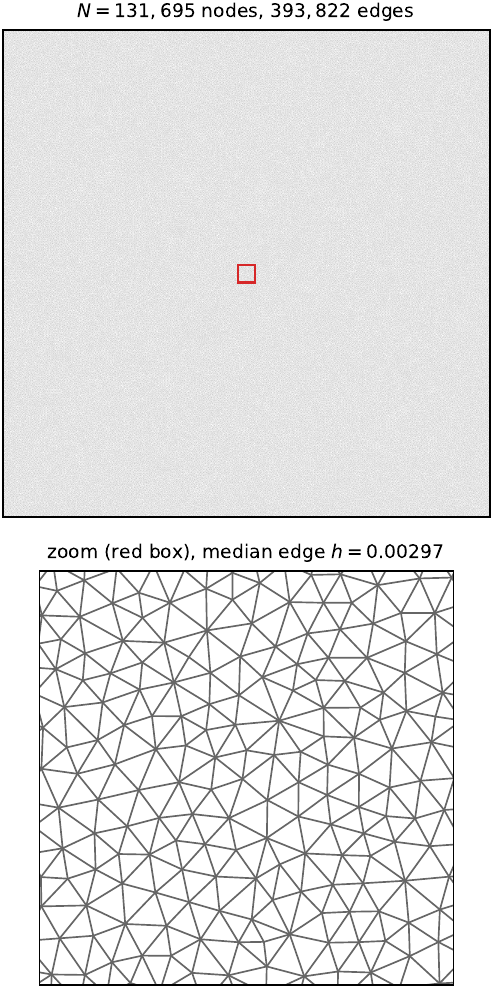}
\caption{Irregular target.}
\end{subfigure}\hfill
\begin{subfigure}[b]{0.47\textwidth}
\centering\includegraphics[width=\linewidth,height=0.21\textheight,keepaspectratio]{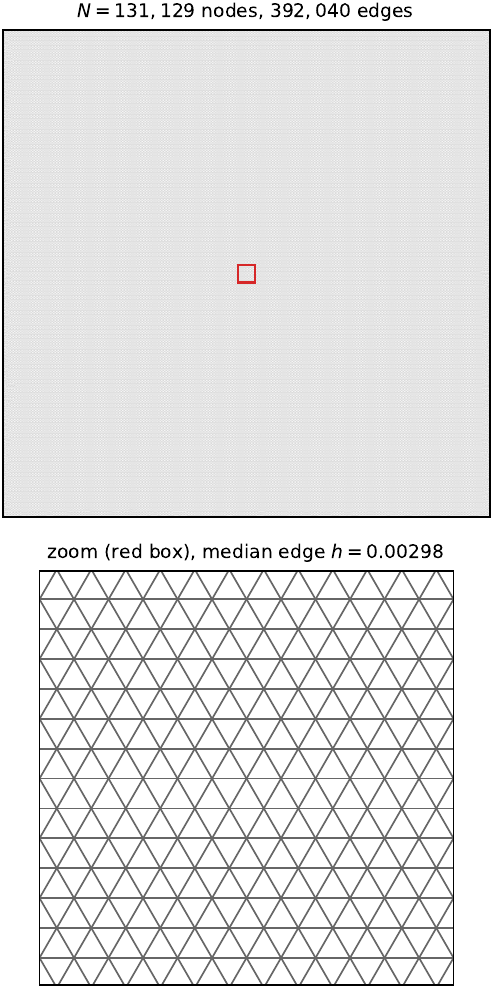}
\caption{Uniform triangular target.}
\end{subfigure}
\caption{The shear-flow target meshes used for tokeniser transfer. Source-grid fields are resampled onto their nodes.}
\label{fig:auto-meshes}
\end{figure}

\begin{figure}[h!]
\centering
\includegraphics[width=0.65\textwidth]{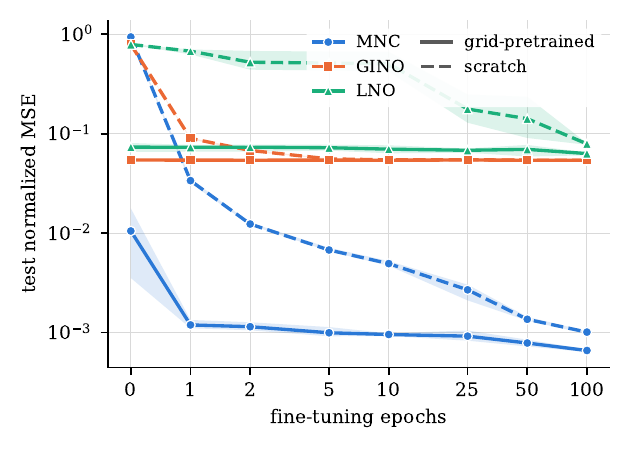}
\caption{Companion to Figure~\ref{fig:transfer-shear} on the uniform triangular lattice. Solid curves use grid-pretrained weights; dashed curves start from random initialisation. Bands show variation across three seeds. Pretraining is not included in the horizontal-axis mesh-training budget.}
\label{fig:transfer-lattice}
\end{figure}

\newpage
\subsection{Cross-domain transfer to native cylinder meshes}
\label{app:cylinder-transfer}
The original reconstruction study also transfers from TRL2D, shear flow and their mixture to cylinder flow (Table~\ref{tab:cylinder-transfer}). Channels are matched by name, and the convolutional stencil is re-instantiated at the target resolution using a characteristic mesh spacing. This differs from within-domain transfer, where the source's physical window is retained. Fine-tuning uses all 1,000 training trajectories, with no warm-up.
\makeatletter\setlength{\@fptop}{0pt}\makeatother
\begin{table}[h!]
  \centering
  \caption{Cross-domain \emph{reconstruction} transfer to native cylinder-flow meshes: normalised MSE after $e$ mesh-training epochs (mean $\pm$ standard deviation, three seeds). $e=0$ is zero-shot; ``scratch'' has no pretraining. The source is CNN for MNC, and each baseline's own tokeniser for GINO and LNO. These are not forecasting results.}
  \label{tab:cylinder-transfer}
  \small
  \setlength{\tabcolsep}{4.5pt}
  \begin{tabular}{llcccc}
    \toprule
    Tokeniser & $e$ & scratch & TRL2D & Shear flow & TRL2D$+$Shear \\
    \midrule
    MNC (ours) & 0   & 1.116 $\pm$ 0.083 & 1.086 $\pm$ 0.161 & 0.999 $\pm$ 0.413 & 0.556 $\pm$ 0.034 \\
               & 25  & \textbf{0.042 $\pm$ 0.001} & \textbf{0.039 $\pm$ 0.000} & \textbf{0.039 $\pm$ 0.002} & \textbf{0.037 $\pm$ 0.002} \\
               & 50  & \textbf{0.035 $\pm$ 0.001} & \textbf{0.031 $\pm$ 0.001} & \textbf{0.031 $\pm$ 0.001} & \textbf{0.030 $\pm$ 0.001} \\
               & 100 & \textbf{0.026 $\pm$ 0.000} & \textbf{0.024 $\pm$ 0.001} & \textbf{0.025 $\pm$ 0.001} & \textbf{0.023 $\pm$ 0.000} \\
    \midrule
    GINO       & 0   & 0.993 $\pm$ 0.034 & \textbf{0.576 $\pm$ 0.013} & \textbf{0.528 $\pm$ 0.008} & \textbf{0.547 $\pm$ 0.008} \\
               & 25  & 0.373 $\pm$ 0.024 & 0.312 $\pm$ 0.014 & 0.362 $\pm$ 0.050 & 0.318 $\pm$ 0.017 \\
               & 50  & 0.281 $\pm$ 0.017 & 0.261 $\pm$ 0.008 & 0.263 $\pm$ 0.010 & 0.270 $\pm$ 0.017 \\
               & 100 & 0.228 $\pm$ 0.011 & 0.233 $\pm$ 0.006 & 0.228 $\pm$ 0.003 & 0.237 $\pm$ 0.013 \\
    \midrule
    LNO        & 0   & 0.969 $\pm$ 0.038 & 1.237 $\pm$ 0.734 & 0.674 $\pm$ 0.017 & 0.661 $\pm$ 0.035 \\
               & 25  & 0.565 $\pm$ 0.002 & 0.327 $\pm$ 0.193 & 0.109 $\pm$ 0.007 & 0.103 $\pm$ 0.006 \\
               & 50  & 0.408 $\pm$ 0.182 & 0.109 $\pm$ 0.006 & 0.088 $\pm$ 0.010 & 0.090 $\pm$ 0.008 \\
               & 100 & 0.091 $\pm$ 0.006 & 0.080 $\pm$ 0.007 & 0.074 $\pm$ 0.012 & 0.071 $\pm$ 0.003 \\
    \bottomrule
  \end{tabular}
\end{table}

These controls limit the claim that MNC transfers best in every sense. GINO has lower zero-shot errors from either individual source; mixed-source pretraining substantially improves MNC's initialisation. After 100 target epochs, MNC has the lowest error, but the benefit of its pretraining over its own scratch run is modest (0.023-0.025 versus 0.026). LNO also benefits from pretraining here. The strong within-physics, changed-discretisation result in Figure~\ref{fig:transfer-shear} should therefore not be generalised to arbitrary cross-physics initialisation or to forecasting.